\documentclass[english,10pt,a4paper]{article}
\usepackage{savesym}
\usepackage{amsmath}
\usepackage{cases}
\usepackage{amssymb, amsthm}
\usepackage[T1]{fontenc}
\usepackage[ansinew]{inputenc}
\usepackage{enumitem}
\usepackage{graphicx}
\usepackage{xspace}
\usepackage{subfigure}
\usepackage{color}
\usepackage[cc]{titlepic}
\usepackage{tabularx}
\usepackage[titletoc,title]{appendix}
\usepackage{multirow}
\usepackage{caption}
\usepackage{titlesec}
\usepackage{url}
\usepackage{mwe}
\usepackage[textfont=it]{caption}
\usepackage{hyperref}
\hypersetup{colorlinks   = true,citecolor    = blue}
\usepackage[numbers,sort]{natbib}
\usepackage{bm}
\usepackage{paralist}

\newtheorem{lem}{Lemma}

\makeatletter

\newcommand{\RZ}{\mathbb{R}}

\newcommand{\diag}{\mathop{\rm diag}\nolimits}
\newcommand{\argmin}{\mathop{\rm argmin}\limits}

\mathchardef\given="626A
\newcommand{\thalf}{{\textstyle\frac12}}
\newcommand{\bbeta}{\mathbf{\beta}}

\newcommand{\bgamma}{\bm{\gamma}}

\newcommand{\Y}{\mathbf{Y}}
\newcommand{\X}{\mathbf{X}}

\titleformat{\section}[runin]
{\normalfont\bfseries}{\thesection.}{0.5em}{}
\titleformat{\subsection}[runin]
{\normalfont\bfseries}{\thesubsection.}{0.5em}{}
\titleformat{\subsubsection}[runin]
{\normalfont\bfseries}{\thesubsubsection.}{1em}{}

\bmdefine\vvarepsilon{\varepsilon}

\title{Incorporating prior information and borrowing information in high-dimensional sparse regression
	using the horseshoe and variational Bayes}
\author{ 	\textbf{Gino B. Kpogbezan}$^{1}$, \textbf{Mark A. van de Wiel}$^{2,3}$,\\
	\textbf{Wessel N. van Wieringen}$^{2,3}$, \textbf{Aad W. van der Vaart}$^{1}$
	\\
	{\small $^1$ Department of Mathematics, Leiden University}
	\\
	{\small Niels Bohrweg 1, 2333 CA Leiden, The Netherlands}
	\\
	{\small $^2$ Department of Epidemiology and Biostatistics, VU University Medical Center}
	\\
	{\small P.O. Box 7057, 1007 MB Amsterdam, The Netherlands}
	\\
	{\small $^3$ Department of Mathematics, VU University Amsterdam}
	\\
	{\small De Boelelaan 1081a, 1081 HV Amsterdam, The Netherlands}
}

\date{}

\begin{document}
	\maketitle
	
	\begin{abstract}
		We introduce a sparse high-dimensional regression approach that can incorporate prior information on the regression 
		parameters and can borrow information across a set of similar datasets. Prior information may for instance 
		come from previous studies or genomic databases, and information borrowed across a set of genes or genomic networks. 
		The approach is based on prior modelling of the regression parameters using the horseshoe prior, with a prior on the sparsity index 
		that depends on external information. Multiple datasets are integrated by applying an empirical Bayes strategy on hyperparameters.
		For computational efficiency we approximate the posterior distribution using a variational Bayes method. The proposed framework is useful for
		analysing large-scale data sets with complex dependence structures. We illustrate this by applications to the reconstruction of gene regulatory
		networks and to eQTL mapping.
	\end{abstract}
	
	
	\section{Introduction}
	\label{intro}
	The analysis of high-dimensional data is important in many scientific areas, and often poses the challenge
	of the availability of a relatively small number of cases versus a large number  of unknown parameters. It has been
	documented both practically and theoretically that under the assumption of sparsity of the underlying model,
	larger effects or dependencies can be inferred even in the very high-dimensional case \cite{Hastie2015,Hahn_Carvalho2015}.
	Still in many cases conclusions can be much improved by incorporating prior knowledge
	in the analysis, or by ``borrowing information'' by simultaneously analysing multiple
	related datasets. In this paper we introduce a methodology that achieves both, and that is 
	at the same time scalable to large datasets in its computational complexity. It is based on 
	an empirical Bayesian setup, where external information is incorporated through the prior, and information is
	borrowed across similar analyses by empirical Bayes estimation of hyperparameters. Sparsity is induced
	through utilisation of the horseshoe prior, and computational efficiency through novel
	variational Bayes approximations to the posterior distribution. We illustrate the methodology by two applications in genomics:
	network reconstruction and  eQTL mapping, but the proposed framework should be useful also for
	analysing other large-scale data sets with complex dependence structures.
	
	Our working model is a collection of linear regression models, indexed by $i=1,\ldots,p$, corresponding
	to $p$ characteristics (e.g.\ genes). For each characteristic we have measurements on $n$ individuals, labelled $j=1,\ldots,n$, consisting
	of a univariate response $Y_i^j$ and a vector $X_i^j$ of $s_i$ explanatory variables. We collect the $n$ responses on
	characteristic $i$ in the $n$-vector $Y_i=(Y_i^1,\ldots,Y_i^n)^T$ and similarly collect the explanatory variables in
	the $n\times s_i$-matrix $X_i$, having rows $X_i^j$, and adopt the regression models
	\begin{equation}
	\label{EqMultR}
	Y_i=X_i\beta_{i}+\epsilon_i, \qquad  i=1, \ldots, p.
	\end{equation}
	Here the regression coefficients $\beta_i$ form a vector in $\RZ^{s_i}$, and the error vectors $\epsilon_i$'s 
	are unobserved.
	The dimension $s_i$ of the regression parameter $\beta_i$ may be different for different characteristics $i$.
	
	Our full set of observations consists of the pairs $(Y_1,X_1),\ldots, (Y_p,X_p)$, whose stochastic dependence
	will not be used and hence need not be modelled. In addition to these regression pairs we assume 
	available prior information on the vectors $\beta_i$ in the form of a 2-dimensional array $P$, whose $i$th row
	presents a grouping of the coordinates of $\beta_i$ into $G$ groups, indexed by $g=1,\ldots, G$: the 
	value $P_{i,t}$ is the index of the group to which the $t$th coordinate of $\beta_i$ belongs. (Because the $\beta_i$ may have
	different lengths, $P$ is a possibly ``ragged array'' and not a matrix.) The information in $P$ is considered to be soft
	in  that coordinates of $\beta_i$ that are assigned to the same group are thought to be similar in size, 
	but not necessarily equal. The information may for instance come from a previous analysis of 
	similar data, or be taken from a genomic database.
	
	We wish to analyse this data, satisfying four aims:
	\begin{compactitem}
		\item Borrow information across the characteristics $i=1,\ldots,p$ by linking the analyses of the models \eqref{EqMultR} for different $i$.
		\item Incorporate the prior information $P$ in a soft manner so that it informs the analysis if correct, but can be
		overruled if completely incompatible with the data.
		\item Allow for sparsity of the explanatory models, i.e.\ focus the estimation towards parameter vectors $\beta_i$ with
		only a small number of significant coefficients, enabling analysis for small $n$ relative to $s_i$ and/or $p$. 
		\item Achieve computational efficiency, enabling analysis with large $s_i$ and/or $p$.
	\end{compactitem}
	To this purpose we model the parameters $\beta_i$ and the scales $\sigma_i$ of the error vectors
	through a prior, and next perform
	empirical Bayesian inference. This analysis is informed by the model \eqref{EqMultR} and the following hierarchy
	of a generating model (referred to as \emph{pInc} later on) for the errors and a prior model  for $(\beta_i,\sigma_i)$:
	\begin{equation}
	\label{sMultRed}
	\begin{aligned}
	\epsilon_i\given \sigma_i&\sim \text{N}(0_n, \sigma_i^2{\bf{I}}_n),\\
	\beta_{i,t}\given \sigma_i, \tau_{i,P_{i,t}},\lambda_{i,t}&\sim  \text{N}\bigl(0,\sigma_i^2 \tau_{i,P_{i,t}}^{2}\lambda_{i,t}^2\bigr),\qquad t= 1, \ldots, s_i,\\
	\sigma^{-2}_i & \sim \Gamma(c, d),\\
	\lambda_{i,t} & \sim C^+(0,1), \qquad t= 1, \ldots, s_i,\\
	\tau_{i,g}^{-2} &  \sim \Gamma(a_g, b_g),  \qquad g=1,\ldots, G.
	\end{aligned}
	\end{equation}
	Here N is a (multivariate) normal distribution,
	${\bf{I}}_n$ is the $(n\times n)$-identity matrix, $C^+(0,1)$ denotes the standard Cauchy distribution restricted to the positive
	real axis, and $\Gamma(u, v)$ denotes the gamma distribution with shape and rate parameters $u$ and $v$.
	As usual the hierarchy should be read from bottom to top, where dependencies of distributions on variables 
	at lower levels are indicated by conditioning, and absence of these variables in the conditioning should be understood as the assumption of conditional
	independence on variables at lower levels of the hierarchy.
	The specification \eqref{sMultRed} gives the model for the $i$th characteristic. The models for different $i$ 
	are linked by assuming the same values of the hyperparameters $a_1,b_1,\ldots, a_G,b_G,c,d$ for all $i=1,\ldots,p$. These
	hyperparameters will be estimated from the combined data $(Y_1,X_1),\ldots, (Y_p,X_p)$ by the
	empirical Bayes method, thus borrowing strength across responses  and achieving the first of the four aims, as listed previously.
	
	We also consider a variant of the model (later referred to as
	\emph{pInc2}) in which the last line of the hierarchy is dropped
	and the parameters $\tau_{i,g}$ are pooled into a single parameter $\tau_{i,g}=\tau_g$ per group
	($i=1,\ldots, p$). The parameters $\tau_g$ are then estimated by empirical Bayes on the data 
	pooled over $i$.  In some of the simulations this model outperformed \eqref{sMultRed}.
	
	The $i$th row of $P$ gives a grouping of the $s_i$ coordinates $\beta_{i,t}$ of $\beta_i$ into $G$
	groups. The scheme \eqref{sMultRed} attaches a latent variable $\tau_{i,g}$ to each group, for
	$g=1,\ldots, G$, whose squares possess inverse gamma distributions, independently across
	groups. These latent variables enter the prior distributions of the coordinates of $\beta_i$, which
	marginally given $\tau_{i,g}$ are scale mixtures of the normal distribution. Choosing the scale
	parameters $\lambda_{i,t}$ from the half-Cauchy distribution gives the so-called \emph{horseshoe
		prior} \cite{Carvalho2009,Carvalho2010}. This may be viewed as a continuous alternative to the
	traditional \emph{spike-and-slab} prior, which is a mixture of a Dirac measure at zero and a widely
	spread second component, and is widely used as a prior that induces sparsity.
	
	The horseshoe density with scale $\tau$ is the mixture of the univariate normal distributions
	$N(0,\tau\lambda)$ relative to the parameter $\lambda\sim C^+(0,1)$. It combines an infinite peak at
	zero with heavy tails, and is able to either shrink parameters to near zero or estimate them
	unbiasedly, much as an improper flat prior. The relative weights of the two effects are moderated by
	the value of $\tau$. In the model \eqref{sMultRed} the coordinates of $\beta_i$ corresponding to the
	same group $g$ receive a common parameter $\tau_{i,g}$, and are thus either jointly shrunk to zero
	or left free, depending on the value of $\tau_{i,g}$. This allows to achieve the aims two and three
	as listed previously. Theoretical work in
	\cite{vdPas2014,vdPasSzabovdVEJS17,vdPasSzabovdV17,Carvalho2010,Datta2013} (in a simpler model) 
	suggests an interpretation of $\tau_{i,g}$ as approximately the fraction of nonzero coordinates in
	the $g$th group, and corroborates the interpretation of $\tau_{i,g}$ as a sparsity parameter. In
	model \eqref{sMultRed} this number is implicitly set by the data, based on the inverse gamma prior
	on $\tau_{i,g}^2$. Requiring the hyperparameters of these gamma distributions to be the same across
	the characteristics $i$ induces the borrowing of information between the characteristics $i$, in
	particular with respect to the sparsity of the vectors $\beta_i$.
	
	Model \eqref{sMultRed} chooses the squares of the scales $\sigma_i$ of the error variables from an
	inverse gamma distribution, which is the usual conjugate prior. The priors on the regression
	parameters $\beta_i$ are also scaled by $\sigma_i$, thus giving them a priori the same order of
	magnitude. This seems generally preferable.
	
	The Bayesian model described by \eqref{EqMultR} and \eqref{sMultRed} leads to a posterior distribution
	of $(\beta_i,\sigma_i)$ in the usual way, but this depends on the hyperparameters $a_1,b_1,\ldots, a_G,b_G,c,d$.
	In Section~\ref{EB} we introduce a method to estimate these hyperparameters from the full data $(Y_1,X_1),\ldots, (Y_p,X_p)$, and next 
	base further inference on the posterior distributions of the parameters $(\beta_i,\sigma_i)$ evaluated at the plugged-in estimates
	of the hyperparameters. Because the prior on the coefficients $\beta_i$ is continuous, the posterior distribution does not
	provide automatic model (or variable) selection, which is a disadvantage of the horseshoe prior relative to the
	spike-and-slab priors. To overcome this, we develop a way of testing for nonzero regression coefficients 
	based on the marginal posterior distributions of the $\beta_{i,t}$  in Section~\ref{varSel}.
	
	The horseshoe prior has gained  popularity, mainly due to its computational advantage over 
	spike-and-slab priors. However, in our high-dimensional setting the approximation of the posterior distribution
	by an MCMC scheme  turns out to be still a computational bottleneck. 
	The algorithm studied by \cite{Bhattacharya2016}, which can be applied in the special case of a single group ($G=1$)
	has complexity $O(n^2s_i)$ for a single regression (i.e. $p=1$) per MCMC iteration.
	We show in Section~\ref{VBvsGibbs} that this is too slow to be feasible in our setting.
	For this reason we develop in Section~\ref{VBa} a variational Bayesian (VB) scheme 
	to approximate the posterior distribution, in order to satisfy the fourth aim in our list.
	
	The variational Bayesian method consists of approximating the posterior distribution by a
	distribution of simpler form, which is chosen as a compromise between computational tractability and
	accuracy of approximation. The quality of the approximation is typically measured by the Kullback-Leibler
	divergence \cite{WainwrightJordan2008}.  Early applications involved standard distributions such as
	Gaussian, Dirichlet, Laplace and extreme value models 
	\cite{Attias1999,McGrory,ArchambeauBach,Armagan09,wand2011}. In the present paper we use nonparametric
	approximations, restricted only by the assumption that the various parameters are (block) independent.
	(This may be referred to as \emph{mean-field} variational Bayes, although this term appears to be used more often for independence of all
	univariate marginals, whereas we use block independence.)
	In this case the variational posterior approximation  can be calculated by iteratively updating the marginal
	distributions \cite{BleiJordan2006,OrmerodWand2010}. Variational Bayes typically produces accurate approximations to posterior means,
	but have been observed to underestimate posterior spread \cite{Giordano2017,Blei2016,CarbonettoStephens2012,MacKay2003,WangTitterington2004,TurnerSahani2011a,WestlingMcCormick2015arXiv,WangBlei2017arXiv}. We find that in our setting the approximations
	agree reasonably well to MCMC approximations of the marginals, although the latter take much longer
	to compute.
	
	The model \eqref{EqMultR}-\eqref{sMultRed} may be useful for data integration in a variety 
	of scientific setups, and for data sources as diverse as gene expression, copy number variations,
	single nucleotide polymorphisms, functional magnetic resonance imaging, or social media data. 
	The external information incorporated in the array $P$ may reflect data of a different type,
	and/or of a different stage of research, and the simultaneous analysis of different characteristics
	allows further data integration. For example, in genetic association studies data from multiple stages 
	can help the identification of true associations \cite{Reif2004,Hamid2009,Hawkins2010}.
	In this paper we consider applications  to gene regulatory networks and to eQLT mapping, which we describe in the next two sections,
	before developing the general algorithms for models \eqref{EqMultR} and \eqref{sMultRed}.
	
	The remainder of the paper is organised as follows. In Section~\ref{VBa} we develop a variational Bayes
	approach to approximate the posterior distributions of the regression parameters for given
	hyperparameters, and show this to be comparable in accuracy to Gibbs sampling  in Section~\ref{VBvsGibbs}, although
	computationally much more efficient. 
	In Section~\ref{EB} we develop the 
	Empirical Bayes (EB) approach for estimating the hyperparameters, and in Section~\ref{varSel} we present a threshold based-procedure for selecting nonzero regression coefficients 
	based on the marginal posterior distributions of the $\beta_{i,t}$.  
	We show in Section~\ref{modSim} by means of model-based simulations that the proposed approach performs better, in terms of both average $\ell_1$-error and average ROC curves, 
	than its ridge counterpart in the framework of network reconstruction. The potential of our approach is shown on real data in Section~\ref{dataApp} both in gene regulatory network reconstruction and in eQTL mapping.
	Section~\ref{end} concludes the paper.

	\section{Network reconstruction}
	\label{App1}
	The identification of gene regulatory networks is crucial for understanding gene function, and hence important for both
	treatment and prediction of diseases.
	Prior knowledge on a given network is often available in the literature, from repositories or pilot studies,
	and combining this with the data at hand can significantly improve the accuracy of reconstruction \cite{Kpogbezan2017}. 
	
	A \emph{Gaussian graphical model} readily gives rise to a special case of the model \eqref{EqMultR}-\eqref{sMultRed}.
	In such a model the data concerning $p$ genes measured in a single individual (e.g.\ tissue) is assumed to form a multivariate Gaussian $p$-vector,
	and the network of interest is the corresponding \emph{conditional independence graph} \cite{whittaker1990graphical}. The nodes
	of this graph are the genes and correspond to the $p$ coordinates of the Gaussian vector. Two nodes/genes
	are connected by an edge in the graph if the corresponding coordinates are \emph{not} conditionally independent given the other coordinates.
	It is well known that this is equivalent to the corresponding element in the precision matrix of the Gaussian
	vector being nonzero \cite{lauritzen1996}.
	
	Assume that we observe a gene vector for $n$ individuals, giving rise to $n$ independent copies $Y^1,\ldots, Y^n$
	of $p$-vectors satisfying 
	\begin{equation}
	\label{EqData}
	Y^j \sim^{\text{iid}}  \text{N}(0_p, \Omega_p^{-1}), \qquad j = 1, \ldots, n.
	\end{equation}
	\noindent Here $\Omega_p$ is the \emph{precision matrix}; its inverse is the covariance matrix of the vector 
	$Y^j$ and is assumed to be positive-definite. The Gaussian graphical model consists of a
	graph with nodes $1,2,\ldots, p$ and with edges $(i,j)$ given by the nonzero elements $(\Omega_p)_{i,j}$ of the precision matrix.
	Hence to reconstruct the conditional independence graph it suffices to determine the non-zero elements of the latter matrix.
	
	We relate this to the notation used in the introduction by writing $Y^j=(Y^j_1,\ldots, Y^j_p)^T$, and next
	collecting the observations $Y^j_i$ per gene $i$, giving the $n$-vector $Y_i=(Y^1_i,\ldots,Y^n_i)^T$, for $i=1,\ldots,p$.
	We next define $$X_i=[Y_1,Y_2,...,Y_{i-1},Y_{i+1},...,Y_p]$$ as the $(n\times (p-1))$-matrix with columns
	$Y_t$, for $t\not=i$. It is well known that the residual when regressing a single coordinate $Y^j_i$
	of a multivariate Gaussian vector linearly on the other coordinates $Y^j_t$, for $t\not=i$, is Gaussian. 
	Furthermore, the regression coefficients $\beta_i=(\beta_{i,t}: t\not=i)$ can be expressed in the precision matrix of $Y^j$ as 
	$$\beta_{i,t}=-\frac{(\Omega_p)_{it}}{(\Omega_p){ii}}.$$
	This shows that \eqref{EqMultR} holds with $s_i=p-1$ and a multivariate normal error vector $\epsilon_i$ with variance $\sigma_i^2$ 
	equal to the residual variance. Moreover, the (non)zero entries in the $i$th row vector of the precision matrix 
	$\Omega_p$ correspond to the (non)zero coordinates of $\beta_i$. Consequently, the
	problem of identifying the Gaussian graphical model can be cast as a variable selection
	problem in the $p$ regression models \eqref{EqMultR}.
	
	This approach of recasting the estimation of the (support of the) precision matrix as 
	a collection of regression problems was introduced by \cite{meinshausen2006}, who employed  Lasso regression
	\cite{Tibshirani96regressionshrinkage,Friedman2008}  to estimate the parameters. Other variable selection methods can be employed as well \cite{Kraemer}.
	A Bayesian approach with Gaussian, ridge-type priors on the regression coefficients was developed in \cite{leday},
	and extended in \cite{Kpogbezan2017} to incorporate prior knowledge on the conditional independence graph.
	A disadvantage of the Gaussian priors  employed in these papers
	is that they are not able to selectively shrink parameters, but shrink them jointly towards zero (although
	prior information used  in \cite{Kpogbezan2017} alleviates this by making this dependent on  prior group).
	This is similar to the shrinkage effect of the ridge penalty \cite{rags2ridgesPaper} relative to the Lasso,
	which can shrink some of the precision matrix elements to exactly zero, and hence possesses intrinsic model selection properties.
	The novelty of the present paper is to introduce the horseshoe prior in order to better model the sparsity of the network. 
	
	We assume that the prior knowledge on the to-be-reconstructed network is available as a ``prior network'', which specifies
	which edges (conditional independencies) are likely present or absent. This information can be coded in an adjacency
	matrix \text{P}, whose entries take the values 0 or 1 corresponding to the absence and presence of an
	edge: $\text{P}_{i,t}=1$ if variable $i$ is connected with variable $t$ and  $\text{P}_{i,t}=0$ otherwise. 
	Thus in this example we only have two groups, i.e.\ $G=2$.
	
	The advantage of reducing the network model to structural equation models of the type
	\eqref{EqMultR} is computational efficiency.  An alternative would be to model the precision matrix
	directly through a prior. This would typically consist of a prior on the graph structure, followed
	by a specification of the numerical values of the precision matrix given its set of nonzero
	coefficients. The space of graphs is typically restricted to e.g.\ decomposable graphs, forests, or trees
	\cite{GiudiciGreen1999,Dobra2004,Jones2005}.
	The posterior distribution of the graph structure can then be used as the basis of inference
	on the network topology. However, except in very small problems, the computational
	burden is prohibitive.

	\section{eQTL mapping}
	\label{App2}
	In eQTL mapping the expression of a gene is taken as a quantitative trait,
	and it is desired to identify the genomic loci that influence it, much as in 
	a classical study of quantitative trait loci (QTL) of a general phenotype. Typically one
	measures the expression of many genes simultaneously and tries to map these to their QTL.
	Since gene expression levels are related to disease
	susceptibility, elucidating these eQLT (expression QTL) may give important insights into the
	genetic underpinnings of complex traits. 
	We shall identify genetic loci here with single nucleotide polymorphisms (SNPs), but other
	biomarkers can be substituted.
	
	Early works by \cite{Cheung2005,Stranger2005,Zhu2008} considered every gene 
	separately for association. However, many genes are believed to be
	co-regulated and to share a common genetic basis \cite{Pujana2007,ZhangHorvath2005}. In addition, SNPs with
	pleiotropic effects may be more easily identified by considering multiple genes together.
	Therefore following \cite{Segal2003,Lee2006,KimXing2012}, we focus on a joint analysis, borrowing
	information across genes. We regress the expression of a given gene on SNPs both within and around the gene,
	where our model is informed about the SNP location.  The sparse parametrization offered by our model is suitable,
	as most  genetic variants are thought to have a negligible (if any) differential effect on expression. 
	
	Suppose we collect the (standardized) expression levels of $p$ genes over $n$ individuals, and identify for each gene $i$ 
	a collection of $s_i$ \text{SNPs} to be investigated for association. For instance, the latter collections 
	might contain all SNPs in a relatively large window around the gene, some of which falling inside the gene and some outside.
	For each individual and SNP we ascertain the number of minor alleles ($0$, $1$ or $2$), and 
	change all 2's to 1's. Because there are not many 2's in the data this does not reduce the information while it simplifies the modelling.
	We use these numbers to form the $n\times s_i$-matrix $X_i$.
	Let $Y_i$ be the $n$-vector of expression levels for gene $i$, and assume the linear model \eqref{EqMultR}.

	It is believed that SNPs that occur within a gene may play a more direct role in
	the gene’s function than SNPs at other genomic locations \cite{SchroderSchumann2005,Lehne2011}.
	Therefore, it is natural to treat \text{SNPs} falling within a given
	gene differently than the ones not falling within that gene. 
	This gives rise to two groups of \text{SNPs} for a given gene, which we can encode
	as prior knowledge in a $2$-dimensional array \text{P} with values $0$ and $1$. 
	
	Thus we have another instance of model \eqref{EqMultR}-\eqref{sMultRed} with two groups, i.e.\ $G=2$.

	\section{ Posterior inference}
	\label{Inference}
	In this section we discuss statistical inference for the model \eqref{EqMultR}-\eqref{sMultRed}.
	This consists of three steps: the approximation to the posterior distribution of the model for given hyperparameters (and given $i$), 
	the estimation of the hyperparameters (across $i$), and finally a method of variable selection.
	
	\subsection{Variational Bayes approximation}
	\label{VBa}
	The \emph{variational Bayes approximation} to a distribution is simply the closest element in a given
	target set $\mathcal{Q}$ of distributions, usually with ``distance'' measured by Kullback-Leibler divergence \cite{WainwrightJordan2008}.
	In our situation we wish to approximate the posterior distribution of the parameter
	$\theta_i:=(\beta_i, \lambda_{i,1}, \cdots, \lambda_{i,s_i},\tau_{i,1}, \cdots, \tau_{i,G}, \sigma_i)$ 
	given $Y_i$ in the model \eqref{EqMultR}-\eqref{sMultRed},  for a fixed $i$. Here we take the regression matrix $X_i$ as given.
	
	Thus the variational Bayes approximation is given as the density $q\in \mathcal{Q}$ that minimizes over $\mathcal{Q}$,
	\begin{equation*}
	KL\bigl(q || p(\cdot\given Y_i)\bigr)  =  \mathbf{E}_q \log \frac{q(\theta_i)}{p(\theta_i\given Y_i)}
	=  \log p(Y_i) - \mathbf{E}_q\log \frac{p(Y_i,\theta_i)}{q(\theta_i)},
	\end{equation*} 
	where $\theta_i\mapsto p(\theta_i\given Y_i)$ is the posterior density, the expectation is taken with respect to $\theta_i$ having the density $q\in \mathcal{Q}$,
	and $(y,\theta_i)\mapsto p(y,\theta_i)=p(y\given \theta_i)\,\pi_i(\theta_i)$ and $y\mapsto p(y)=\int p(y, \theta_i)\,d\theta_i$ 
	are the joint density of $(Y_i,\theta_i)$ and the marginal density of $Y_i$, respectively,  in the model
	\eqref{EqMultR}-\eqref{sMultRed}, with prior density $\pi_i$ on $\theta_i$. As the marginal density is free of $q$, minimization
	of this expression is equivalent to maximization of the second term
	\begin{equation}
	\label{LowerBound}
	\mathbf{E}_q\log \frac{p(Y_i,\theta_i)}{q(\theta_i)}.
	\end{equation}
	By the non-negativity of the Kullback-Leibler divergence, this  expression is a lower bound 
	on the logarithm of the marginal density $p(Y_i)$ of the observation. For this reason it is 
	usually referred to as ``the lower bound'', or ``ELBO'',  and 
	solving the variational problem  is equivalent to maximizing this lower bound.
	
	The set $\mathcal{Q}$ is chosen as a compromise between computational tractability and accuracy of approximation.
	Restricting $\mathcal{Q}$ to distributions for which all marginals of $\theta_i$ are independent 
	is known as \emph{mean-field} variational Bayes, or also as the ``na\"ive factorization'' \cite{WainwrightJordan2008}.
	Here we shall use the larger set of distributions under which  the blocks of $\beta$, $\lambda$, $\tau$ and $\sigma$-parameters
	are independent. Thus we optimize over probability densities  $q$ of the form 
	$$q(\theta_i)=
	q_{\beta}(\beta_i)\cdot q_{\lambda}(\lambda_{i,1}, \cdots, \lambda_{i,s_i})\cdot q_{\tau}(\tau_{i,1}, \cdots, \tau_{i,G}) \cdot  q_{\sigma}(\sigma_i).$$
	There is no explicit solution to this optimization problem. However, if all marginal factors but a single one in
	the factorization are fixed, then the latter factor can be characterised easily, using the non-negativity of the
	Kullback-Leibler divergence. This leads to an iterative algorithm, in which the factors are updated in turn.
	
	In the Supplementary Material (SM) we show that in our case the iterations take the form:
	\begin{equation}
	\label{margPost}
	\begin{aligned}
	\beta_i\given Y_i &              \sim \text{N}\big(\beta_i^* ,\Sigma_i^*\big),\\
	\lambda_{i,t}\given Y_i    &   \sim \Lambda_{\lambda_{it}},    \qquad \qquad t=1, \cdots, s_i,\\
	\tau_{i,g}^{-2}\given Y_i &  \sim \Gamma(a_{i,g}^*, b_{i,g}^*),    \quad g=1, \cdots, G,\\
	\sigma^{-2}_i\given Y_i &   \sim \Gamma\big(c_{i}^*, d_{i}^*\big),
	\end{aligned} 
	\end{equation}
	where $\Lambda_{l}$ is the distribution with probability density function proportional to
	$$\lambda \mapsto \frac{1}{\lambda(1+\lambda^2)}e^{-l\lambda^{-2}},\qquad (\lambda>0),$$
	and the parameters on the right hand side satisfy
	\begin{align*}
	\Sigma_i^* &= \Big[\mathbf{E}_{q_{\sigma}^*}(\sigma_i^{-2})\big(X_i^{T}X_i + {\bf{D}}^{-1}_{\mathbf{E}_{q^*_{\tau_i}\cdot q^*_{\lambda_i}} }\big)\Big]^{-1},\\
	\beta_i^*   &= \big(X_i^{T}X_i + {\bf{D}}^{-1}_{\mathbf{E}_{q^*_{\tau_i}\cdot q^*_{\lambda_i}} }\big)^{-1}X_i^{T}Y_i,\\
	a_{i,g}^*        &=a_g+0.5\cdot\frac{s^g_{i}}{2},\\
	b_{i,g}^*        & =b_g +0.5\cdot \mathbf{E}_{q_{\sigma}^*}(\sigma_i^{-2})\mathbf{E}_{q_{\textbackslash\tau_g}^*}\Big({\beta^g_{i}}^{T}{\bf{D}}^{-1}_{\lambda_i}\beta^g_{i}\Big),  \qquad g=1, \cdots, G,\\
	c_{i}^*        &=c+\frac{n}{2}+\frac{s_i}{2},\\
	d_{i}^*       & =d +0.5\cdot \mathbf{E}_{q_{\textbackslash\sigma}^*}\Big(\beta_i^{T}{\bf{D}}^{-1}_{\tau_i\lambda_i}\beta_i\Big)+ 0.5\cdot \mathbf{E}_{q_{\beta}^*}(Y_i- X_i\beta_i)^{T}(Y_i-X_i\beta_i),\\
	{\bf{D}}_{\lambda_i}  &=\text{diag}(\lambda_{i,1}^2, \ldots ,\lambda_{i,s_i}^2), \\
	{\bf{D}}_{\tau_i\lambda_i}  &=\text{diag}(\tau_{i,P_{i,1}}^{2}\lambda_{i,1}^2, \ldots ,\tau_{i,P_{i,s_i}}^{2}\lambda_{i,s_i}^2), \\
	{\bf{D}}^{-1}_{\mathbf{E}_{q^*_{\tau_i}\cdot q^*_{\lambda_i}} } &= \diag\Big(\mathbf{E}_{q_{\tau_i}^*}(\tau_{i,P_{i,1}}^{-2}) \mathbf{E}_{q_{\lambda_{i1}}^*}(\lambda_{i,1}^{-2}), \ldots ,\mathbf{E}_{q_{\tau_i}^*}(\tau_{i,P_{i,s_i}}^{-2}) \mathbf{E}_{q_{\lambda_{is_i}}^*}(\lambda_{i,s_i}^{-2})\Big),\\
	l_{it} &= \frac{1}{2}\mathbf{E}_{q_{\sigma}^*}(\sigma_i^{-2})\mathbf{E}_{q_{\tau}^*}(\tau_{i,P_{i,t}}^{-2})\mathbf{E}_{q_{\beta}^*}(\beta_{i,t}^2).
	\end{align*}
	In these expressions, $s^g_{i}$ is the number of $g$'s in the $i$-th row of the 2-dimensional array
	\text{P} encoding the $G$ groups, $g=1, \cdots, G$;  and
	$\beta_i^g =\{\delta_{\{\text{P}_{i,r}=g\}} \beta_{i,r}: r \in \{1, \cdots, s_i \} \}$ 
	is the vector obtained from $\beta_i$ by replacing the coordinates not corresponding to group $g$ by $0$.
	
	The expected value of $z_{it}:=(\lambda_{it})^{-2}$, which appears in the expression of $\beta_i^*$, $\Sigma_i^*$, $b_{i,g}^*$ and $d_{i}^* $ above, 
	is given in the following lemma.
	
	\begin{lem}
		The norming constant for $\Lambda_l$ is $2 \exp(-l)/E_1(l)$ and the expectation of $z_{it}=(\lambda_{it})^{-2}$ if
		$\lambda_{it}\sim \Lambda_{\lambda_{it}}$ is given by 
		$$ \mathbf{E}(z_{it}) =  \frac{1}{l_{it}\cdot \exp(l_{it})\cdot E_1(l_{it})} - 1,$$
		where $E_1$ is the \emph{exponential integral function of order 1}, defined by
		$$E_1(x) \equiv \int_{x}^{\infty}\frac{e^{-t}}{t}dt, \qquad x \in \mathbb{R}^+.$$
	\end{lem}
	
	\begin{proof}
		This follows by easy manipulation and the standard density transform formula.
	\end{proof}
	
	\noindent The function $E_1$ can be evaluated effectively by the function {\texttt{expint\_E1()}} in the R package {\bf{gsl}} \cite{Hankin2007}. The latter uses the GNU Scientific Library \cite{Galassi2009}. 
	
	In addition, the variational lower bound \eqref{LowerBound} on the log marginal likelihood at $q=q^*$ takes the form (See SM for details)	
	\begin{equation}
	\label{LowerBoundConcrete}
	\begin{aligned}
	\mathcal{L}_i &= -\frac{n}{2}\log (2\pi) - s_i\log (\pi) + \frac{1}{2} \log |\Sigma_i^\ast| + \frac{1}{2} s_i \\
	&\qquad+\sum\limits_{g = 1}^G ( a_g \log b_g - \log\Gamma(a_g) - a_{i,g}^\ast \log b_{i,g}^\ast + \log\Gamma(a_{i,g}^\ast) ) \\
	&\qquad  + c\log d - \log\Gamma(c) - c_i^\ast\log d_i^\ast + \log\Gamma(c_i^\ast)\\
	&\qquad + \sum\limits_{g = 1}^G \Bigl( \frac{1}{2} \mathbf{E}_{q_{\sigma}^*}(\sigma_i^{-2}) \mathbf{E}_{q_{\tau}^*}(\tau_{i,g}^{-2}) \mathbf{E}_{q^*}({\beta^g_{i}}^T{\bf{D}}^{-1}_{\lambda_i}\beta^g_{i}) \Bigr) \\
	&\qquad  + \sum\limits_{t = 1}^{s_i} \Bigl(\log E_1(l_{it}) + \frac{1}{\exp(l_{it})E_1(l_{it})} \Bigr).
	\end{aligned}
	\end{equation}

	\subsection{Global Empirical Bayes}
	\label{EB}
	Model \eqref{sMultRed} possesses the $G+1$ pairs of  hyperparameters $(a_1,b_1),\cdots,(a_G,b_G),(c,d)$.
	The pair $(c,d)$ controls the prior of the error variances $\sigma_i^2$; we fix this 
	to numerical values that render a vague prior, e.g.\ to $(0.001, 0.001)$. In contrast, we let the values
	of the parameters  $\alpha=(a_1,b_1, \cdots,  a_G,b_G)$ be determined by the data. As these hyperparameters
	are the same in every regression model $i$, this allows information to be borrowed across the regression equations,
	leading to  \emph{global shrinkage} of the regression parameters. 	The approach is similar to the one in \cite{vanDeWiel}.
	
	Precisely, we consider the criterion
	\begin{align}
	\label{EqEB}
	\alpha=(a_1,b_1, \cdots,  a_G,b_G) &\mapsto \sum_{i=1}^p \mathbf{E}_{q}\log \frac{p_\alpha(Y_i,\theta_i)}{q(\theta_i)}\\
	&=\sum_{i=1}^p\mathbf{E}_{q}\log \frac{p(Y_i\given\theta_i)}{q(\theta_i)}+\sum_{i=1}^p\mathbf{E}_{q}\log \pi_\alpha(\theta_i).\nonumber
	\end{align}
	The maximization of the function on the right with respect to $q\in\mathcal{Q}$ for fixed $\alpha$ 
	leads to the variational estimator $q^*$ considered in Section \ref{VBa}
	(which depends on $\alpha=(a_1,b_1, \cdots,  a_G,b_G)$).
	Rather than running the iterations \eqref{margPost} for computing this estimator to ``convergence'', next inserting
	$q=q^*_\alpha$ in the preceding display \eqref{EqEB},  and finally maximizing the resulting expression with respect to $\alpha$,
	we blend iterations to find $q^*$ and $\alpha^*$ as follows. Given an \emph{iterate} $q^*$ of \eqref{margPost} we set $q$ in \eqref{EqEB} equal to $q^*$ and find its maximizer $\alpha^*$
	with respect to $\alpha$. Next given $\alpha^*$ we set $\alpha$ (in the display following
	\eqref{margPost}  equal to $\alpha^*$ and use \eqref{margPost} 
	to find a next iterate of $q^*$. 
	We repeat these alternations to ``convergence''.
	
	For fixed $q=q^*$ the far right side in the second row of the preceding display depends on $\alpha$ only through
	$$ \sum_{i=1}^p\mathbf{E}_{q^*} \bigl(\log \pi_\alpha(\theta_i)\bigr).$$
	
	\noindent Using the approximation $\log(x)-\frac{1}{2x}  \approx  \Psi(x)=\frac{\partial }{\partial x}\log \Gamma (x)$, where $\Psi$ is the digamma function, the maximization yields (see SM for details)
	\begin{align*}
	\hat{a}_g & \approx  \thalf\Bigl[\log \Bigl(\sum_{i=1}^p\mathbf{E}_{q^*} \tau_{i,g}^{-2} \Bigr)
	- p^{-1}\Bigl(\sum_{i=1}^p\mathbf{E}_{q^*} \log \tau_{i,g}^{-2}\Bigr) - \log p \Bigr]^{-1} \\
	\hat{b}_g & = \hat{a}_g\cdot p\cdot \Bigl[\sum_{i=1}^p\mathbf{E}_{q^*} \tau_{i,g}^{-2} \Bigr]^{-1}
	\end{align*}
	where $g \in \{1, \cdots, G\}$.
	The following algorithm summarizes the above described procedure.
	\vspace{0.5cm}
	
	\begin{tabular}{|l|} 
		\hline
		\bf{Variational algorithm with sparse local-global shrinkage priors } \\  \hline
		\small		1: {\textbf{Initialize}} \\
		\small	     $   a^{(0)}_g=b^{(0)}_g=10^{-3}$, $g \in \{1, \cdots, G\}$ and $ \forall i\in \mathcal{I}$, $b_{i,g}^*=d_i^*= 10^{-3}$,         $\epsilon=10^{-3}$, \\
		M $=10^3$  and $k=1$ \\ 
		\small	      2: \textbf{while} $\max|\mathcal{L}_i^{(k)}-\mathcal{L}_i^{(k-1)}|\geq \epsilon$ \textbf{and} $2\leq k \leq$M \textbf{do}\\
		\small\quad \quad \quad E-step: Update variational parameters\\
		\small      3:\quad \quad \quad \textbf{for} $i=1$ to $p$ \textbf{update}\\
		\small\quad \quad \quad \quad   $ a^{*(k)}_{i,g}$, $ c^{*(k)}_{i}$,  \\
		\small\quad \quad \quad \quad   $\Sigma_i^{*(k)}$, $\beta_i^{*(k)}$, $b^{*(k)}_{i,g}$, $d_{i}^{*(k)}$, $l_{it}^{(k)}$ and $\mathcal{L}_i^{(k)}$; \quad $\forall g$ and $\forall t$ in that order\\
		\small \quad\quad \quad \quad \textbf{end for}\\
		\\
		\small\quad \quad \quad M-step: Update hyperparameters\\
		\small      4:\qquad 
		\small  	$a_g^{(k)}$, $b_g^{(k)}$; \quad $\forall g$ \\
		\small      5:\quad \quad \quad$ k\leftarrow k+1$\\
		\small	    6: \textbf{end while}
		\\  \hline
	\end{tabular}

	\subsection{Variable selection}
	\label{varSel}
	
	Because the horseshoe prior is continuous, the resulting posterior distribution does not set parameters
	exactly equal to zero, and hence variable selection requires an additional step.
	We investigated two schemes that both take the marginal posterior distributions of the parameters as input.
	
	\subsubsection{Thresholding}
	\label{varSel1}
	A natural method is to set a parameter $\beta_{i,r}$ equal to zero (i.e.\ remove the corresponding
	independent variable from the regression model) if the point 0 is not in the tails of its marginal
	posterior distribution, or more precisely, if 0 does belong to a central marginal credible
	interval for the parameter.  Given that our variational Bayes scheme produces conditional Gaussian
	distributions, this is also equivalent to the absolute ratio of posterior mean and standard
	deviation
	\begin{equation}
	\label{EqDefkappa}
	\kappa_{i,r}=\frac{\left\lvert \mathbf{E}_{q^{i \ast}}\left[\beta_{i,r}|\mathbf{Y}_i\right]\right\rvert }{ \mathbf{sd}_{q^{i \ast}}\left[\beta_{i,r}|\mathbf{Y}_i\right]}
	\end{equation}
	not exceeding some threshold. (In the network setup of Section~\ref{App1} we use the symmetrized
	quantity $(\kappa_{i,r}+\kappa_{r,i})/2$, as the two constituents of the average refer to the same parameter.)
	
	To determine a suitable cutoff or credible level we applied the variational Bayes procedure of
	Section~\ref{VBa} with all credible levels $\gamma$ on a grid with step size 5\% within the range $[10\%, 99.99\%]$,
	resulting in a model, or set of `nonzero' parameters $\beta_{i,r}$, for every $\gamma$.  
	We allow rather lenient credible levels because the model might benefit from the inclusion of fewer variables, in particular when strong collinearity is present.
	We next
	refitted the model \eqref{EqMultR}-\eqref{sMultRed} with the non-selected parameters $\beta_{i,r}$
	set equal to 0, evaluated the variational Bayes lower bound on the likelihood \eqref{LowerBound}
	(equivalently \eqref{LowerBoundConcrete}), and chose the value of $\gamma$ that maximized this
	likelihood and the corresponding model.
	When refitting we did not estimate the hyperparameters ($a$'s and $b$'s for \emph{pInc},
	$\tau$'s for \emph{pInc2}, as explained in Section~\ref{EB}), but used the values resulting from the
	entire data set.  Even though this procedure sounds involved, it is computationally fast, because it
	is free of the empirical Bayes step and typically needs to evaluate only models with few predictors.
	
	
	\subsubsection{An alternative selection scheme}
	\label{varSel2}
	As an alternative  selection scheme we investigated the \emph{decoupled shrinkage and selection} (DSS) 
	criterion proposed by \cite{Hahn_Carvalho2015}. For each regression model $i$,
	given the posterior mean vector $\bar{\bbeta}_i=\mathbf{E}_{q^{i \ast}}\left[\beta_{i,\cdot}|\mathbf{Y}_i\right]$ determined 
	by the pooled procedure of Sections~\ref{VBa}-\ref{EB}, this calculates the adaptive lasso type estimate
	\begin{equation}\label{dss}
	\hat{\bgamma}_i(\lambda_i) = \argmin_{\bgamma_i}\Bigl[\frac 1n\|\X_i\bar{\bbeta}_i - \X_i\bgamma_i\|^2_2 + \lambda_i\sum_{t=1}^p \frac{|\gamma_{it}|}{|\bar\beta_{i,t}|}\Bigr],
	\end{equation}
	and next chooses the model corresponding to the nonzero coordinates of $\bgamma_i$.  The authors
	\cite{Hahn_Carvalho2015} advocate this method over thresholding, in particular because it may
	better handle multi-collinearity. In genomics applications, such as the eQTL Example (Section~\ref{dataApp2}),
	multi-collinearity is likely strong, in particular between neighbouring genomic locations. 
	Another attractive aspect of (\ref{dss}) is that it only relies on the posterior means, which we have
	shown to be accurately estimated by the variational Bayes approximation.  
	
	In the DSS approach the thresholding in order to obtain models of different sizes is performed through the smoothing parameters
	$\lambda_i$. The authors \cite{Hahn_Carvalho2015}
	propose a heuristic to choose $\lambda_i$ based on the credible interval of the explained
	variation. An alternative is to apply $K$-fold cross-validation based on the squared prediction error:
	\begin{equation}\label{mselambda}
	\text{MSE}(\lambda_i) = \frac1n\sum_{k=1}^K \|\Y_i^k - \X_i^k\hat{\bgamma}_i^{-k}(\lambda_i)\|^2_2,
	\end{equation}
	where superscript $k$ refers to the observations used as test sample in fold $k=1,\ldots, K$, 
	and $-k$ to the complementary training sample used to calculate $\hat{\bgamma}_i^{-k}(\lambda_i)$,
	by \eqref{dss} with $\X_i^{-k}$ and $\bar{\bbeta}_i^{-k}$ replacing $X_i$ and $\bar\bbeta_i$.
	Again we throughout fix the hyperparameters of the priors  to the ones resulting from the variational Bayes algorithm on the entire data set.
	We have found that the function $\lambda_i\mapsto\text{MSE}(\lambda_i)$ can be flat, 
	which, to some extent, is a `by-product' of the strong shrinkage properties of horseshoe prior.
	(Given a sparse true vector, many posterior means $\bar{\beta}_{i,r}$
	will be close to zero, which renders the DSS solution (\ref{dss}) less dependent on $\lambda_i$.)
	To overcome this, and because we prefer sparser models, we used the maximum value 
	of $\lambda_i$ for which the MSE is within 1 standard error of the minimum of the mean square errors.
	
	In the next sections, if not specified, selection should be understood as the first scheme based on thresholding.

	\section{Simulations}
	\label{modSim}
	
	We performed model-based simulations to compare model \eqref{sMultRed}, referred to as \emph{pInc},
	with the alternative method \emph{pInc2}, in which there is only one parameter $\tau_g$ 
	per group, and their ridge counterpart \emph{ShrinkNet} (\cite{leday}). 
	We refer to the latter paper for comparisons of \emph{ShrinkNet} to other competing methods.
	\emph{ShrinkNet} was indeed shown in \cite{leday} to outperform the \emph{graphical lasso} \cite{Friedman2008}, the \emph{SEM Lasso} \cite{meinshausen2006} and the \emph{GeneNet} \cite{Schaefer2006} using exactly the same simulated data below.
	As \emph{ShrinkNet} was developed for network reconstruction only and does not
	incorporate prior knowledge, we initially considered the setup of network reconstruction in Section~\ref{App1}
	and set $G=1$ in \eqref{sMultRed}.  Next we compared \emph{pInc} and \emph{pInc2} in the same
	network recovery context, but incorporating prior information.  Finally, we compared the accuracy and
	computing time of our variational Bayes approximation approach with Gibbs sampling-based strategies
	\cite{Bhattacharya2016}. 
	
	
	\subsection{Model-based simulation}
	\label{modSim1}
	
	We generated data $Y^1,\ldots, Y^n$ according to \eqref{EqData}, for $p = 100$ and
	$n \in \{10, 100, 200, 500\}$ to reflect high and low-dimensional designs.  We generated precision
	matrices $\Omega_p$ corresponding to {\emph{band}}, {\emph{cluster}} and {\emph{hub}} network topologies
	\cite{Zhao2012,leday} from a G-Wishart distribution \cite{MohammadiWit2015} with scale matrix equal to the
	identity and $b=4$ degrees of freedom. 
	
	The performance of the methods was investigated using average $\ell_1$ errors
	$\|\hat{\beta}_0 - {\beta} _0 \|_1$ and $\|\hat{\beta}_1 - {\beta} _1 \|_1$ across $50$
	replicates of the experiment. Here $\beta_1$ (or $\beta_0$) is the vector consisting of all nonzero (or zero) values of 
	the partial correlation matrix $-{(\Omega_p)_{it}}/{(\Omega_p){ii}}$ except the diagonal elements,
	and $\hat{\beta}_1$ (or $\hat{\beta}_0$) is the vector consisting of the corresponding 
	posterior means.
	
	\begin{table}
		\begin{center}				
			\begin{tabular}{|c | c | c | c | c |}
				\hline
				& Sample size	  & {{\emph{ShrinkNet}}}  & {{\emph{pInc2}} } & {{\emph{pInc}} } \\ \hline
				\multirow{4}{*}{Band}  & $n=10$      & 25.26    &  1.77      & 0.66  \\  
				&  $n=100$   & 265.89  &  180.42  & 78.46      \\  
				&  $n=200$   & 291.33  &  113.12  & 121.29   \\  
				&  $n=500$   & 251.47  &   81.38   & 150.62     \\	\hline
				
				\multirow{4}{*}{Cluster}  & $n=10$    & 15.74   & 0.71      & 0.51    \\  
				&  $n=100$  & 224.89 & 186.88  & 39.97    \\ 
				&  $n=200$  & 259.94 & 130.70  & 98.77  \\ 
				&  $n=500$  & 231.33 &   82.82  & 107.58 \\	\hline
				
				\multirow{4}{*}{Hub}  & $n=10$      & 7.44      & 0.28    & 0.34  \\  
				&  $n=100$   & 155.87  & 8.70    & 47.85  \\  
				&  $n=200$   & 154.63  & 12.65  & 84.46   \\   
				&  $n=500$   & 132.50  & 21.51  & 106.31     \\	\hline
				
			\end{tabular}
		\end{center}
		\caption{Average $l_1$ error, $\|\hat{\beta}_0 - {\beta} _0 \|_1$ across $50$ simulation replicates with sample size $n\in \{10, 100, 200, 500\}$ and $p=100$. The precision matrices used correspond respectively to \it{Band}, \it{Cluster} and \it{Hub} structure.}\label{L1error0}
	\end{table}
	
	\begin{table}
		\begin{center}				
			\begin{tabular}{|c | c | c | c | c |}
				\hline
				& Sample size	  & {{\emph{ShrinkNet}}}  & {{\emph{pInc2}} } & {{\emph{pInc}} } \\ \hline
				\multirow{4}{*}{Band}  & $n=10$      & 220.15  & 220.55 & 221.92  \\  
				&  $n=100$   & 162.58  & 112.01 & 134.82  \\ 
				&  $n=200$   & 124.01  & 66.08   &   65.66    \\
				&  $n=500$   &  72.51   & 29.08   &   29.25    \\	\hline
				
				\multirow{4}{*}{Cluster}  & $n=10$    & 288.86  & 288.64 & 289.44  \\  
				&  $n=100$  & 254.03  & 160.05 & 217.48  \\  
				&  $n=200$  & 215.88  &   75.24 & 86.54  \\   
				&  $n=500$  & 133.22  &  27.99  & 29.95    \\	\hline
				
				\multirow{4}{*}{Hub}  & $n=10$     & 40.25  & 39.34 &  40.52   \\  
				&  $n=100$   & 24.14  & 15.39 &  13.99   \\  
				&  $n=200$   & 17.58  &   9.42 &  8.65   \\   
				&  $n=500$   & 12.54  &   5.42 &  5.26    \\	\hline
				
			\end{tabular}
		\end{center}
		\caption{Average $l_1$ error, $\|\hat{\beta}_1 - {\beta} _1 \|_1$ across $50$ simulation replicates with sample size $n\in \{10, 100, 200, 500\}$ and $p=100$. The precision matrices used correspond respectively to \it{Band}, \it{Cluster} and \it{Hub} structure.}\label{L1error1}
		
	\end{table}
	
	The results are displayed in Tables~\ref{L1error0} and~\ref{L1error1}.
	Both methods \emph{pInc} and \emph{pInc2} outperform \emph{ShrinkNet} in all simulation setups. 
	For the nonzero parameters (`signals')  \emph{pInc} and \emph{pInc2} are on par,
	but for the zero parameters \emph{pInc} outperforms \emph{pInc2} for 
	small $n$ in the Band and Cluster topologies,  but when $n$ increases and in the Hub topology
	this turns around. 
	
	Somewhat worrisome is that the performance of all methods on the zero parameters initially seems
	to suffer from increasing sample size $n$. The empirical Bayes choice of
	shrinkage level clearly favours strong shrinkage for small $n$, giving good performance
	on the zero parameters, but relaxes this when the sample size increases. Thus the better
	performance for increasing $n$ on the nonzero parameters is partly offset by 
	a decline in performance on the zero parameters. 
	This balance between zero and nonzero parameters 
	is restored only for relatively large sample sizes. A similar
	phenomenon was observed in \cite{vdeWiel18}.
	
	Tables~\ref{L1error0n} and~\ref{L1error1n} compare the performance of \emph{pInc} and
	\emph{pInc2} when prior information is available (both with sample size $n=10$). The prior information
	consists either of the correct adjacency matrix $P$ for the network (i.e.\ $P_{i,t}=1$ if $\Omega_{i,t}\not=0$
	and $P_{i,t}=0$ otherwise), or an adjacency matrix in which 50 \% of the positive entries are
	correct. The latter matrix was obtained by swapping a random selection of half the 1s in the
	correct adjacency matrix with a random selection of equally many 0s. 
	The tables shows that \emph{pInc} usually outperforms \emph{pInc2},
	the zero parameters in the \emph{Hub} case with $50\%$ true edge prior knowledge being the
	only significant exception.
	
	\begin{table}
		\begin{center}				
			\begin{tabular}{|c | c | c | c | c |}
				\hline
				& Quality of prior Info &	 {{\emph{pInc2}} } & {{\emph{pInc}} } \\ \hline
				\multirow{2}{*}{Band}  & True model      & 6.90  & 0.68  \\  
				&  50\% true edge info   & 6.66  & 5.30   \\  \hline
				
				\multirow{2}{*}{Cluster}  & True model                & 4.96 & 0.60    \\  
				&  50\% true edge info  & 3.25  & 3.28    \\  \hline
				
				\multirow{2}{*}{Hub}  & True model                & 0.22 & 0.27   \\  
				&  50\% true edge info & 0.46  & 5.88    \\  \hline
				
			\end{tabular}
		\end{center}
		\caption{Average $l_1$ error, $\|\hat{\beta}_0 - {\beta} _0 \|_1$ across $50$ simulation replicates with sample size $n = 10$ and $p=100$. Qualities of prior information correspond to true model and $50\%$ true edge information.}\label{L1error0n}
	\end{table}

	\begin{table}
		\begin{center}				
			\begin{tabular}{|c | c | c | c |}
				\hline
				& Quality of prior Info &	 {{\emph{pInc2}} } & {{\emph{pInc}} } \\ \hline
				\multirow{2}{*}{Band}  & True model               & 216.25  & 209.48    \\  
				&  50\% true edge info   & 219.57  & 217.39    \\  \hline
				
				\multirow{2}{*}{Cluster}  & True model          & 285.72  &  281.21   \\  
				&  50\% true edge info  & 286.98  &  286.73  \\  \hline
				
				\multirow{2}{*}{Hub}      & True model                & 29.40 &   27.55  \\  
				&  50\% true edge info   & 37.79  &   34.60  \\  \hline
				
			\end{tabular}
			\caption{Average $l_1$ error, $\|\hat{\beta}_1 - {\beta} _1 \|_1$ across $50$ simulation replicates with sample size $n =10$ and $p=100$. Qualities of prior information correspond to true model and $50\%$ true edge information.}\label{L1error1n}
		\end{center}
	\end{table} 
	
	To study the performance of the different methods on model selection we computed ROC curves, showing the
	true positive rate (TPR) and false positive rate (FPR) as a function of the threshold on the test
	statistic \eqref{EqDefkappa} for inclusion of a parameter in the model.  
	Figure~\ref{rocP} shows that in the absence of prior information \emph{pInc2} performs best, closely followed
	by \emph{pInc}, and both methods outperform \emph{ShrinkNet}.
	Given either correct or 50\% correct information 
	\emph{pInc} is the winner, as seen in  Figure~\ref{rocP2}, which also
	shows the usefulness of incorporating prior information.
	These findings are consistent with the results on estimation presented in 
	Tables~\ref{L1error0}--\ref{L1error1n} in their ordering of \emph{pInc} above
	\emph{pInc2} in the case of availability of external information.
	
	\begin{figure}
		\begin{center}
			\includegraphics[scale=0.6]{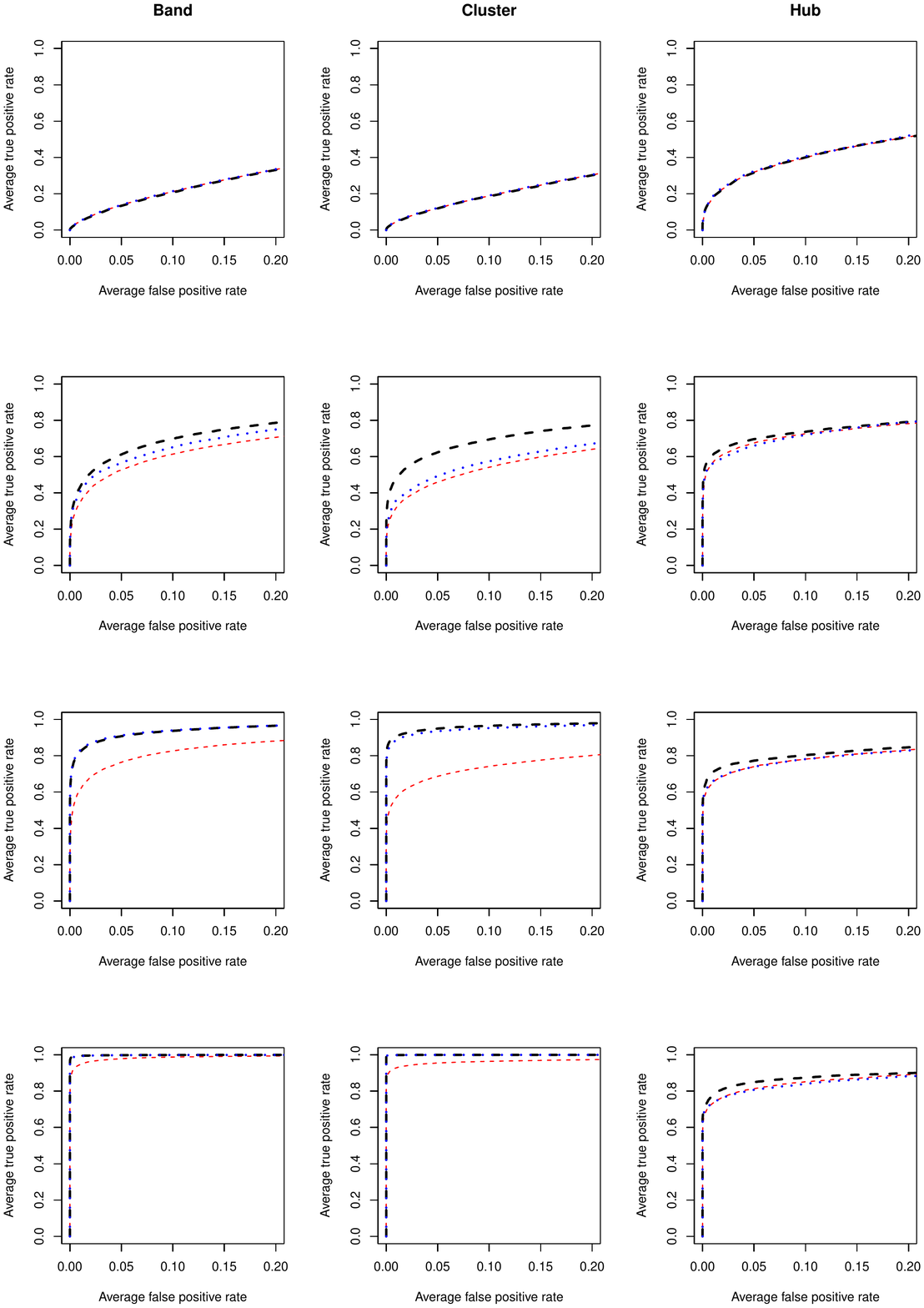}
		\end{center}
		\caption{Average partial-ROC curves comparing performance of ShrinkNet (dashed red), pInc2 (dashed black) and pInc (dashed blue) where $n \in \{10, 100, 200, 500\}$ and $p=100$. First, second, third and fourth rows correspond respectively to the performances of $n=10$, $n=100$, $n=200$ and $n=500$. }\label{rocP}
	\end{figure}
	
	\begin{figure}[ht]
		\begin{minipage}[c]{0.33\linewidth}
			\centering
			\includegraphics[width=1.0\textwidth]{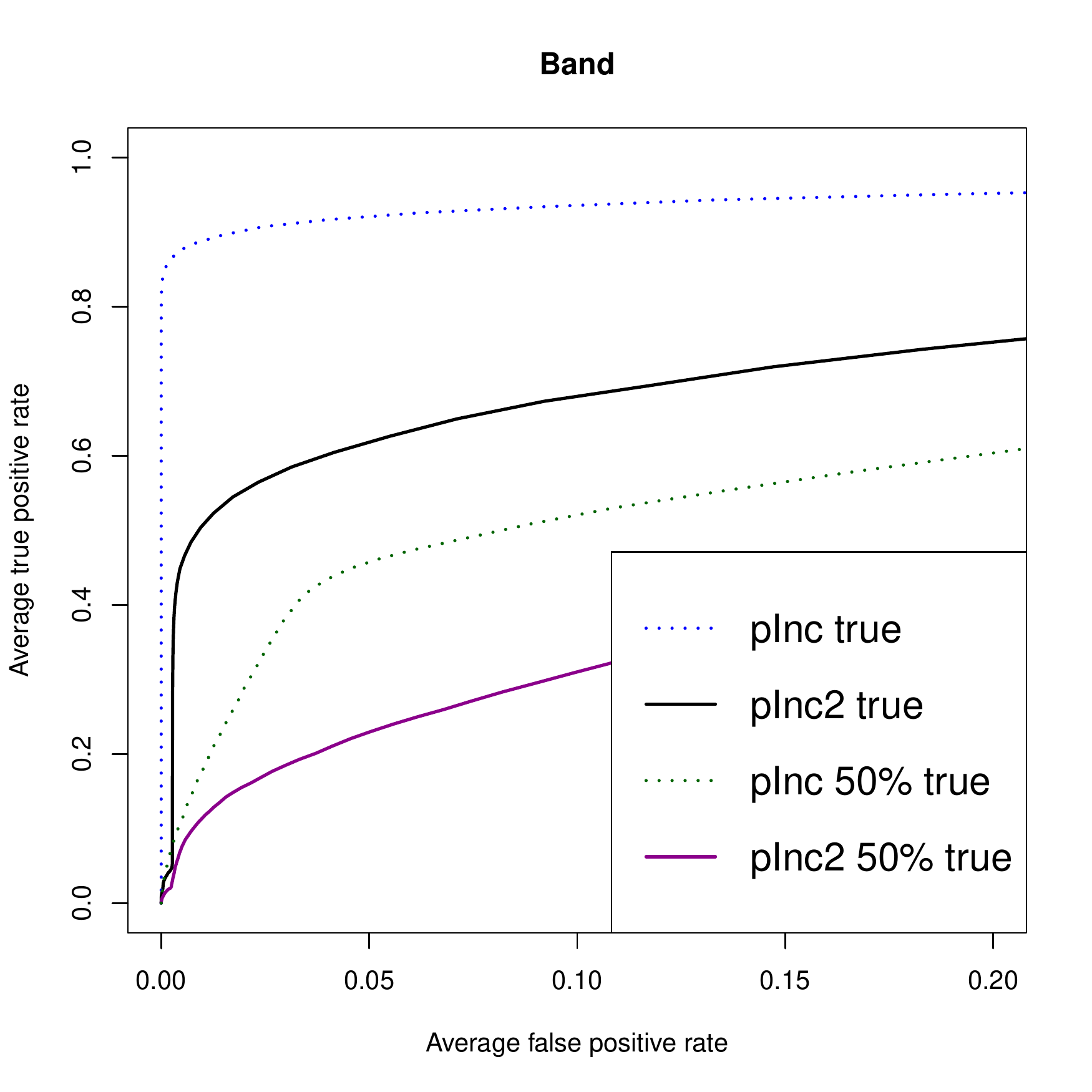}
		\end{minipage}\hfill
		\begin{minipage}[c]{0.33\linewidth}
			\centering
			\includegraphics[width=1.0\textwidth]{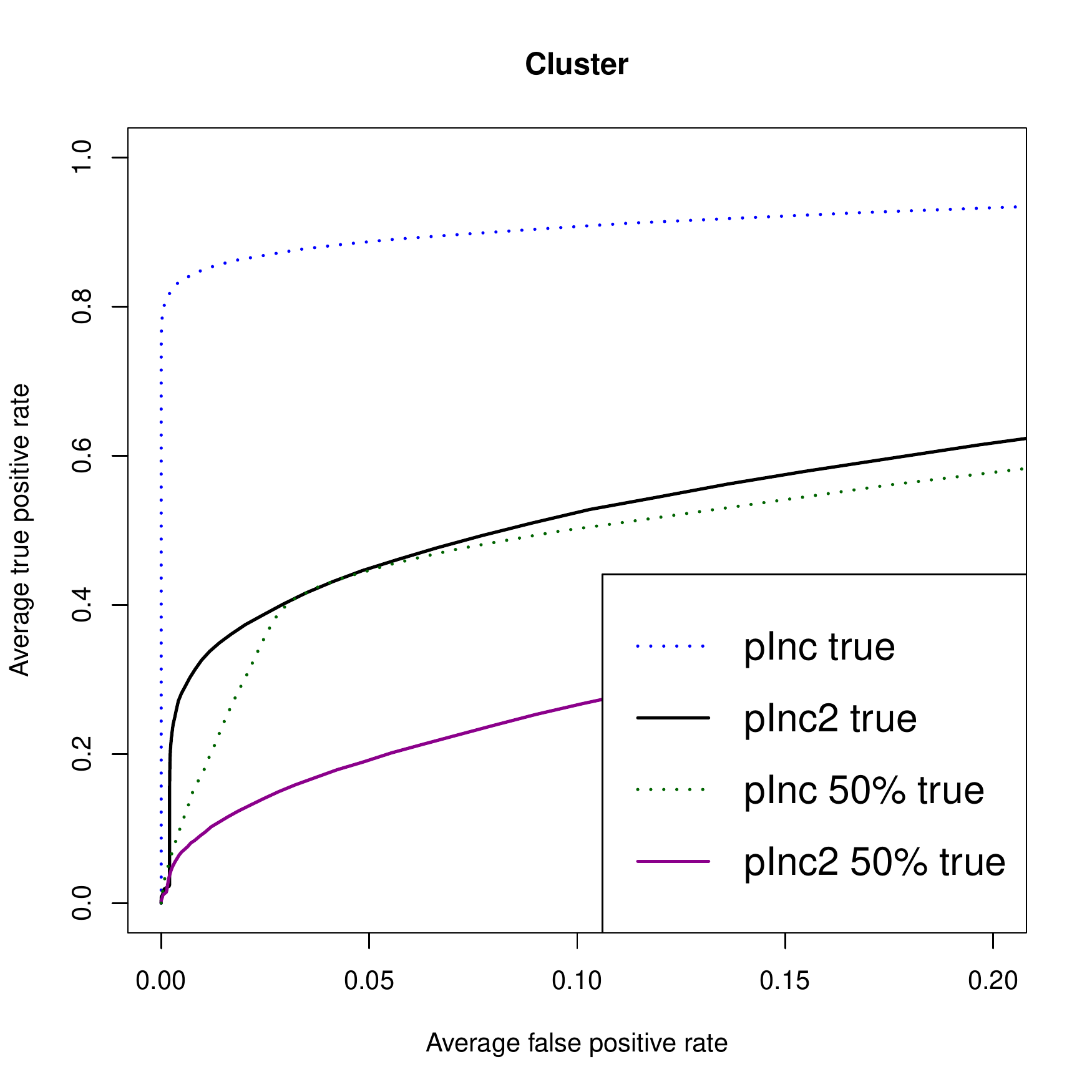}
		\end{minipage}\hfill
		\begin{minipage}[c]{0.33\linewidth}
			\centering
			\includegraphics[width=1.0\textwidth]{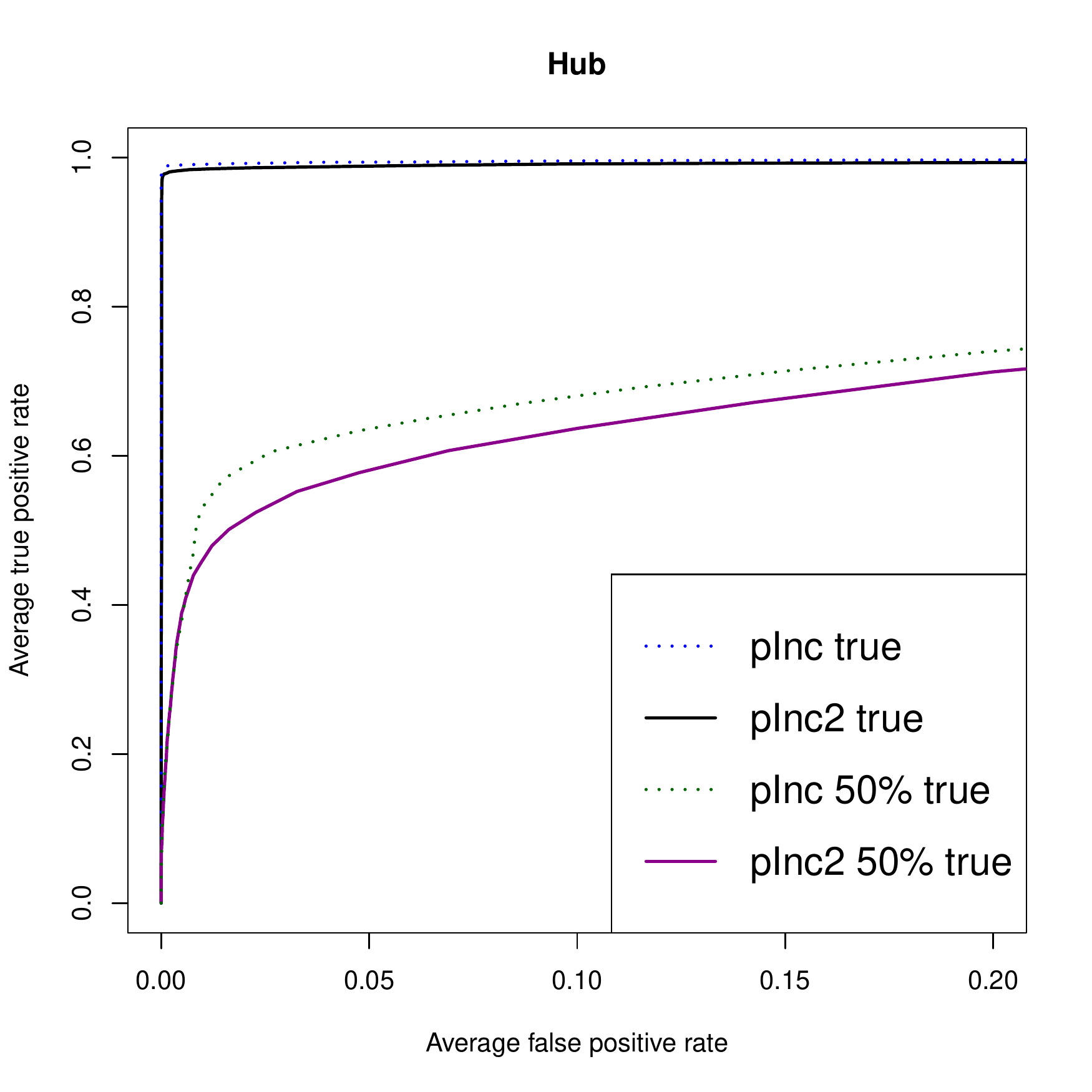}
		\end{minipage}\hfill
		\caption{Average partial-ROC curves comparing performance of pInc using perfect prior information
			(dashed blue), pInc2 using perfect prior information (black), pInc using $50\%$ true edge information (dashed dark green) and pInc2 using $50\%$ true edge information (darkmagenta). Sample size and network dimension 
			were $n=10$ and $p=100$.}\label{rocP2}
	\end{figure}

	Figure~3 in the supplementary material displays histograms of the EB estimates of prior parameter/hyperparameter $\tau^2$'s by pInc (TauSq) and pInc2 (TauSq2) across the $50$
	simulation replicates. The initial hyperparameter value for \emph{pInc2} was set to $0.05$.  The
	figure shows that the estimated parameters are bigger (hence less shrinkage) when the sample size is
	larger. Furthermore, for a fixed sample size the estimates are reasonably stable, the quotient of
	the largest and smallest across the 50 replicates being below a small constant.

	\subsection{Variational Bayes vs MCMC}
	\label{VBvsGibbs}
	
	We investigated the quality of the variational approximation by comparing it 
	to the output of a long MCMC run. As we only use the univariate marginal
	posterior distributions of the regression parameters for model selection, we focused on these.
	We ran a simulation study with a single regression equation (say $i=1$) with
	$n=p=100$, and compared the variational Bayes estimates of the marginal densities with the
	corresponding MCMC-based estimates. We sampled $n=100$ independent replicates from a
	$p=100$-dimensional normal distribution with mean zero and $(p\times p)$-precision matrix
	$\Omega_p$, and formed the vector $Y_1$ and matrix $X_1$ as indicated in Section \ref{App1}.  The
	precision matrix was chosen to be a \emph{band matrix} with lower and upper bandwidths equal to 4,
	thus a band of total width 9.  For both the variational approximations and the MCMC method we used 
	prior hyperparameters $c=d=0.001$ and prior hyperparameters $(\hat{a}, \hat{b})$ 
	(resp. $\hat{\tau}^2$ for {\emph{pInc2}})
	fixed to the values set by the \emph{global} empirical Bayes method described
	in Section \ref{EB}.  The MCMC iterations were run $nIter=4\times10^4$ times without thinning,
	after which the first $nBurnin=2\times10^4$ iterates were discarded \cite{Polson2014}.  
	Tables~\ref{vbGibbs} and~\ref{vbGibbs2} summarize the comparison.
	
	The correspondence between the two methods is remarkably good. 
	The posterior means obtained from the variational method 
	are even slightly better as estimates of the true parameters than the ones from the MCMC method,
	in terms of $\ell_1$-loss. With respect to computing time the variational method was
	vastly superior to the MCMC method, which would hardly be feasible even for $n=p=100$.  
	
	\begin{table}
		\begin{tabular}{|c|c|c|} 
			\hline
			&   \vtop{\hbox{\strut average $l_1$ loss $||\hat{\beta}_1-\beta_1||_1$ }\hbox{\strut in 20 replications ($i=1$) }}     &  \vtop{\hbox{\strut computing time needed }\hbox{\strut for all the 100 regressions}}  \\  \hline
			{\emph{pInc}}        &  1.41    &    58 sec       \\ \hline
			
			MCMC method	&  2.22    &    13h 15 min      \\\hline
		\end{tabular}
		\captionof{table}{Performance comparison between pInc and the MCMC method.}\label{vbGibbs}
		\vspace{0.5cm}
		
		\begin{tabular}{|c|c|c|} 
			\hline
			&   \vtop{\hbox{\strut average $l_1$ loss $||\hat{\beta}_1-\beta_1||_1$ }\hbox{\strut in 20 replications ($i=1$) }}     &  \vtop{\hbox{\strut computing time needed }\hbox{\strut for all the 100 regressions}}  \\  \hline
			{\emph{pInc2}}     &  2.25    &    1 min 48 sec    \\ \hline
			
			MCMC method	 &  3.03    &    13h 19 min       \\\hline
		\end{tabular}
		\caption{Performance comparison between pInc2 and the MCMC method.}\label{vbGibbs2}
	\end{table}

	\section{Applications}
	\label{dataApp}
	
	We applied the methods to two real datasets, both as illustration.

	\subsection{Reconstruction of the apoptosis pathway}
	\label{dataApp1}
	
	The cells of multicellular organisms possess the ability to die by a process called programmed cell
	death or \emph{apoptosis}, which contributes to maintaining tissue homeostasis.
	Defects in the apoptosis-inducing pathways can eventually lead to expansion of a population of
	neoplastic cells and  cancer \cite{Hanahan2000,Krammer1998,Igney2002}.  
	Resistance to apoptosis may increase the escape of tumour cells from surveillance by the immune
	system. Since chemotherapy and irradiation act primarily by inducing apoptosis,
	defects in the apoptotic pathway can make cancer cells resistant to therapy. For this reason
	resistance to apoptosis remains  an important clinical problem. 
	
	In this section we illustrate the power of our method in reconstructing the apoptosis network from
	lung cancer data \cite{landi} from the Gene Expression Omnibus (GEO). The data comprises $p = 84$
	genes, consisting of $n_1=49$ observations from normal tissue and $n_2=58$ observations from tumor
	tissue, hence $n=107$ observations in total.  We fitted \emph{pInc} on the tumor data, using the
	data on normal tissue as prior knowledge. To the latter aim we fitted \emph{pInc} to the normal data
	with a single group $G=1$, and applied the model selection procedure of Section~\ref{varSel1} to
	create an array $P$ of incidences, which served as input when fitting \emph{pInc} on the tumor
	data. The idea is that, while tumors and normal tissue may differ  strongly in terms of mean
	gene expression, the gene-gene interaction network may be relatively more stable.
	
	When fitting the \emph{pInc} model with the two groups (gene interaction absent or present in normal
	tissue), we observed a huge difference in the empirical Bayes estimates of the hyperparameters
	governing the priors of the parameters $\tau^{-2}$ of the two groups, namely prior mean
	${\hat{a}_0}/{\hat{b}_0}= 8476.97$ for absent and ${\hat{a}_1}/{\hat{b}_1}=3.70$ for present in
	the prior network.  This strongly indicates the relevance of the prior knowledge
	\cite{Kpogbezan2017}, so that superior performance of \emph{pInc} in the reconstruction can be
	expected.
	
	Figure~\ref{spApopt} displays the reconstructed undirected network by \emph{pInc} with Selection procedure \ref{varSel1}.
	A total number of 27 edges were found with various edge strengths.
	The ten most significant edges in decreasing order were:
	PRKACG $\leftrightarrow$ FASLG,
	MYD88 $\leftrightarrow$ CSF2RB,
	PIK3R2 $\leftrightarrow$ CHUK,
	TNFRSF10B $\leftrightarrow$ CHP1,
	PRKAR1B $\leftrightarrow$ AKT2,
	PIK3R2 $\leftrightarrow$ NGF,
	TRAF2 $\leftrightarrow$ BAX,
	TNF $\leftrightarrow$  IL1B,
	PRKAR2B $\leftrightarrow$ AKT3, and
	TRAF2 $\leftrightarrow$ PIK3R2.
	%
	
	Node degrees varied from 0 to 4 with PIK3R2 and PRKAR1A yielding the highest degree 4,
	followed by TRAF2 having degree 3, and CHUK, CHP1, BIRC3, FAS, IL1B and NFKBIA having each degree 2. 
	
	\begin{figure}
		\includegraphics[width=.8\textwidth]{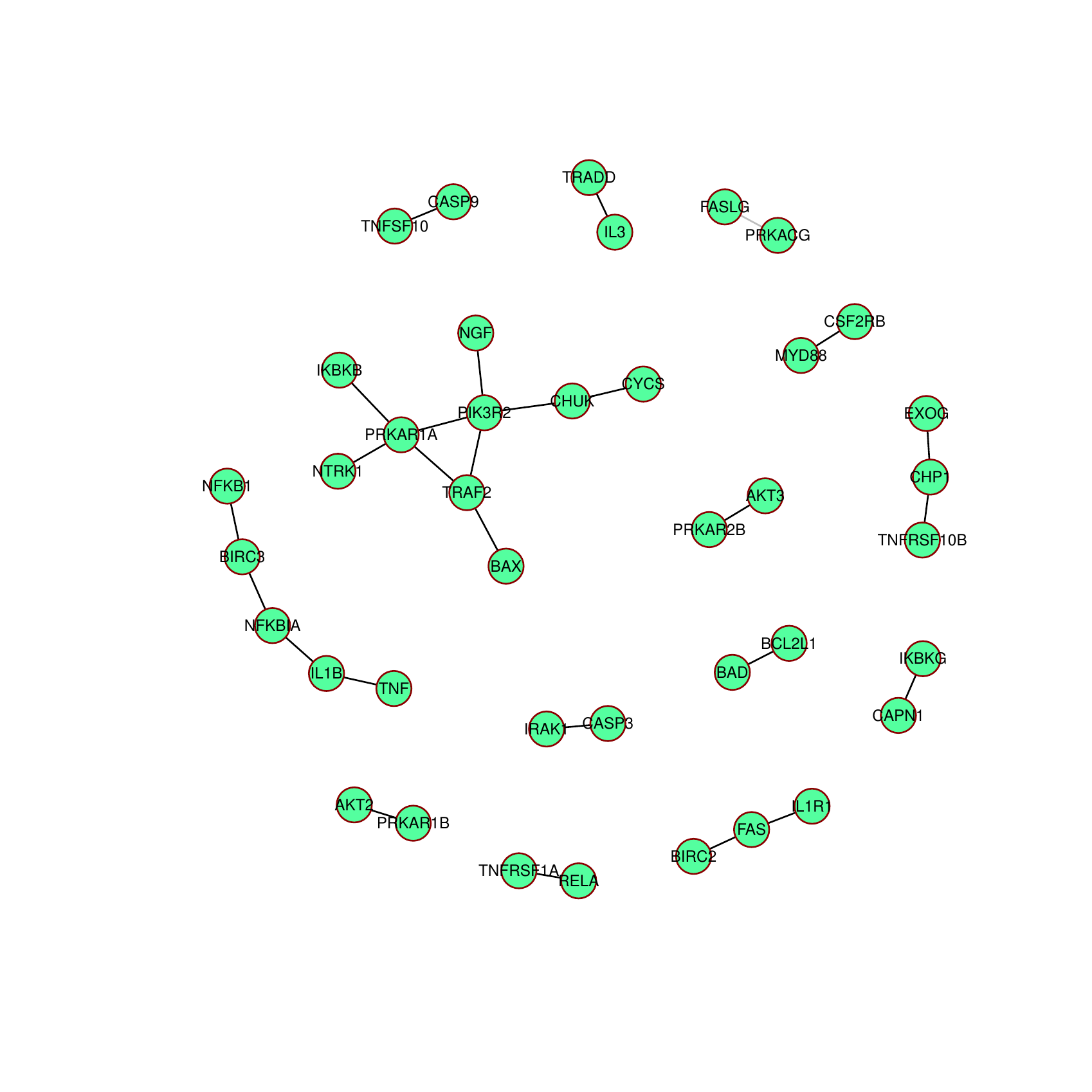} 
		\caption{Apoptosis network reconstructed for the 84 genes by pInc.}\label{spApopt}	
	\end{figure}

	\subsection{eQTL mapping of the p38MAPK pathway}
	\label{dataApp2}
	The p38MAPK pathway is activated \emph{in vivo} by environmental stress and inflammatory
	cytokines, and plays a key role in the regulation of inflammatory cytokines biosynthesis.  Evidence
	indicates that p38MAPK activity is critical for normal immune and inflammatory response
	\cite{Lee1994,Berenson2006,Hwang1997}.  The pathway also plays an important role in cell
	differentiation. Its key role in the conversion of myoblasts to differentiated myotubes during
	myogenic progression has been established by \cite{Li2000,Wu2000,Zetser1999}. More recently,
	\emph{in vivo} studies demonstrated that p38MAPK signalling is a crucial determinant of myogenic
	differentiation during early embryonic myotome development \cite{deAngelis2005}. Finally, the 
	pathway is involved in chemotactic cell migration \cite{Heit2002,Heuertz1999}.
	Lack of p38MAPK function may lead to cell cycle deficiency and
	tumorigenesis, and genetic variants of some genes in the p38MAPK pathway are associated
	with lung cancer risk \cite{Feng2018}.  Studying the pathway in healthy cells may enhance
	understanding the underlying biological mechanism, but has received less attention.
	
	We investigated the association between single nucleotide polymorphisms (SNPs) and the
	genes in the P38MAPK pathway, using GEUVADIS data.  In the GEUVADIS project
	\cite{Lappalainen2013}, 462 RNA-Seq samples from lymphoblastoid cell lines were obtained, while
	the genome sequence of the same individuals is provided by the 1000 Genomes Project. The samples in this
	project come from five populations: CEPH (CEU), Finns (FIN), British (GBR), Toscani (TSI) and Yoruba
	(YRI).  In our analysis we excluded the YRI population samples and samples without expression and
	genotype data, which resulted in a remaining sample size of 373. We also excluded SNPs with minor
	allele frequency (MAF) $< 5\%$. Using a window of $10^5$ bases upstream and $10^5$ downstream of
	every gene, we obtained a total number of 42,054 SNPs for the 99 genes of the pathway belonging to the
	22 autosomes. This resulted in a system of 99 regression models, with dimensions varying from 56 to
	1169. 
	We scaled (per gene) the gene expression data prior to the computations.
	
	Following Section~\ref{App2} we classified the SNPs connected to each gene as located either within the
	gene range or outside, and applied \emph{pInc} with two groups ($G=2$).  
	We observed a big difference in the empirical Bayes estimates of the 
	hyperparameters of the priors of  $\tau^{-2}$:  mean value ${\hat{a}_0}/{\hat{b}_0}= 27,568.76$
	for SNPs outside the gene ranges versus ${\hat{a}_1}/{\hat{b}_1}= 4102.46$ for
	SNPs inside. The prior information is thus clearly relevant, and hence an improved mapping
	by \emph{pInc} can be expected.
	
	We found using Selection procedure \ref{varSel1} the expression levels of 13 out of the 99 genes 
	(genes 15, 40, 48, 50, 51, 61, 75, 78, 85, 86, 93, 96, 98) to be associated 
	with a total number of 50 SNPs from the 42,054 SNPs under consideration. Gene 50
	yielded the highest number 9 of associated SNPs,  followed by gene 40 with 6 SNPs
	and genes 86, 93 and 96 with 5 SNPs each.  
	Figures~\ref{spBetas1} and~\ref{spBetas2} display the estimates of the effect sizes of the SNPs
	(posterior means $\mathbf{E}_{q^*}(\beta_{i,r}\given Y_i)$), green for SNPs outside
	the gene ranges and blue for SNPs within a gene, with `red stars' indicating
	the SNPs that were selected by the selection procedure presented in Section~\ref{varSel1}.
	The 6 largest associations were
	observed within genes 93, 15, 96, 98 and 78 (red vertical lines in Figures~\ref{spBetas1} and~\ref{spBetas2}).
	The active SNPs for all genes, except genes 40 and 50 
	(although for gene 50 only one of the selected SNPs is not within),
	are located inside the gene range.
	This confirms the belief that SNPs falling inside genes are more prone 
	to influence these genes than SNPs outside.
	
	The SNP effects on the remainder 86 ($=99-13$) genes are similar to the ones on gene 1 displayed in Figure~\ref{spBetas2}. The selection obtained by using {\emph{pInc}}-DSS is similar (see SM).

	\begin{figure}
		\includegraphics[width=\textwidth, height=0.95\textheight]{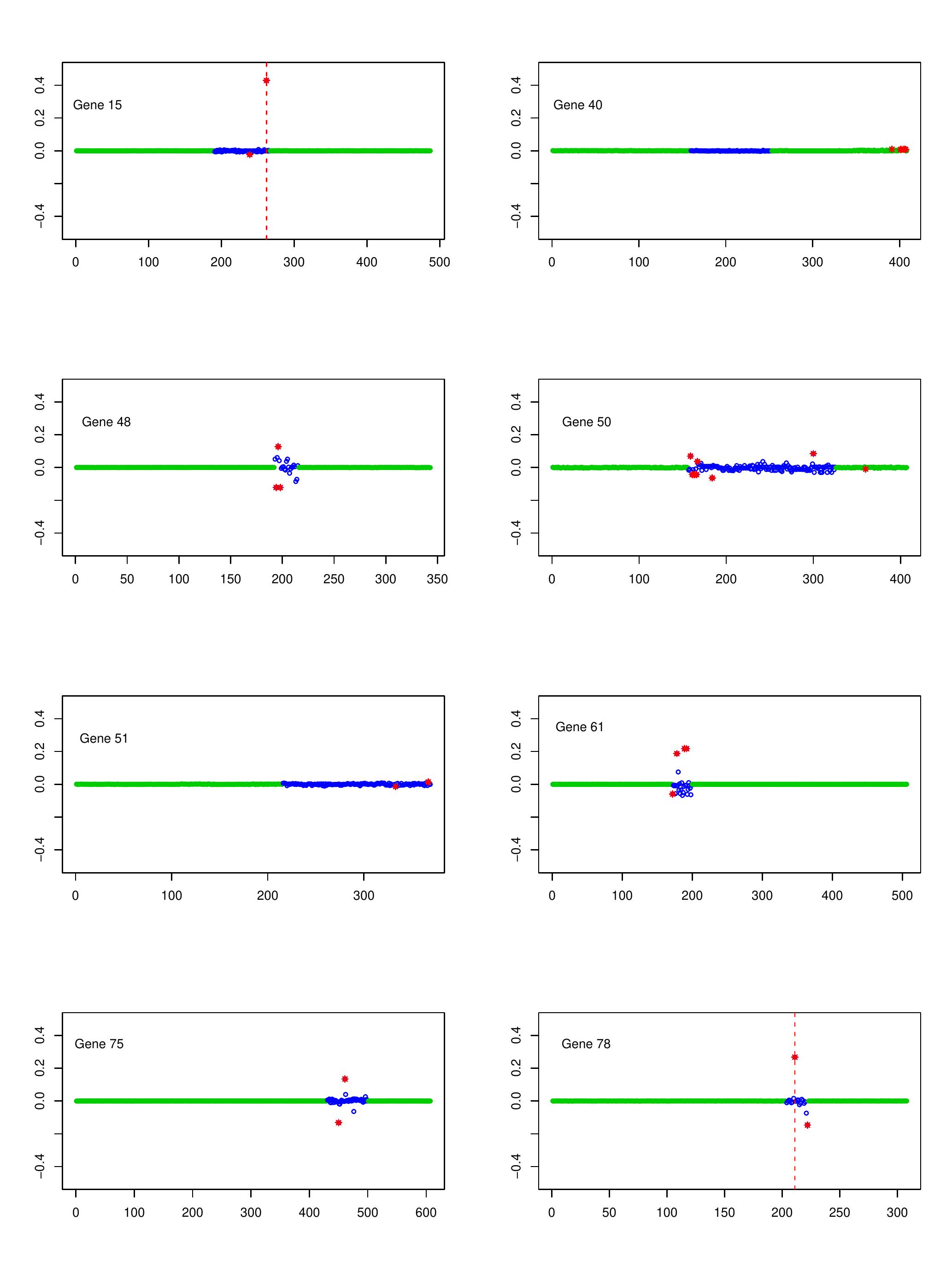} 
		\caption{Estimates of SNP effects on genes 15, 40, 48, 50, 51, 61, 75 and 78 using {\emph{pInc}}. Green dots indicate effects estimates for SNPs outside the gene range and blue dots for SNPs inside the gene range. Red `stars' indicate selected SNP effects.}\label{spBetas1}	
	\end{figure}
	\begin{figure}
		\includegraphics[width=\textwidth, height=0.8\textheight]{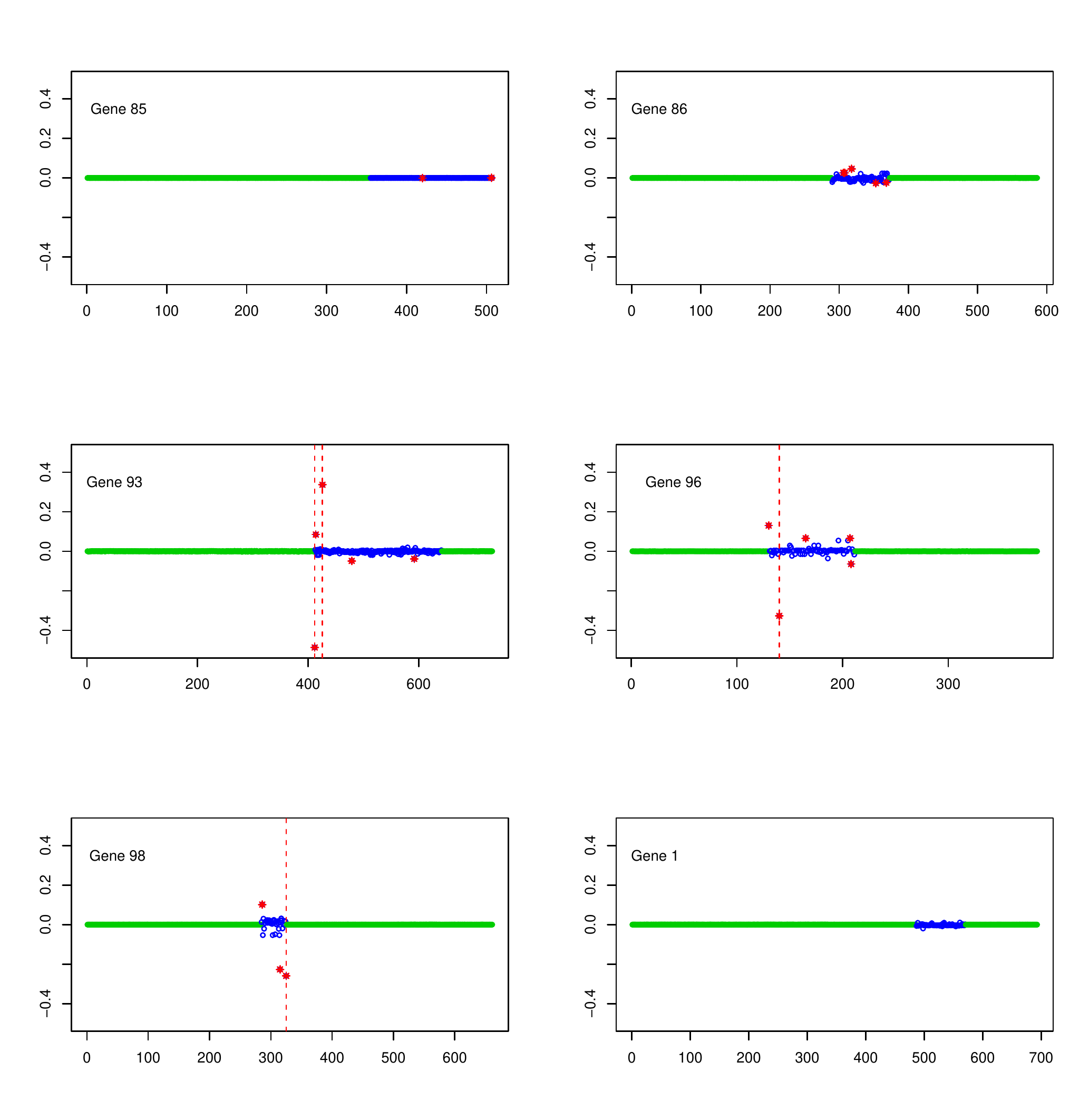} 
		\caption{Estimates of SNP effects on genes 75, 78, 85, 86, 93, 96, 98 and 1 using {\emph{pInc}}. Green dots indicate effects estimates for SNPs outside the gene range and blue dots for SNPs inside the gene range. Red `stars' indicate selected SNP effects.}\label{spBetas2}	
	\end{figure}
	
	\subsubsection{Comparison of {\emph{pInc}}-DSS with lasso}
	
	From the many dedicated methods for eQTL analysis 
	\cite{CaiX2011,LiG2018,Segal2003,Lee2006,KimXing2012}, we chose the lasso
	as a bench-mark to compare the model selection by \emph{pInc} combined with \emph{DSS} \ref{varSel2}. 
	Our choice for DSS comes from the interest to investigate whether '\emph{pInc} + lasso' indeed outperforms a direct lasso, as suggested for the basic horseshoe.  
	As a criterion we used predictive performance when using a sparse model restricted to include a
	maximal number of predictor variables (SNPs).  As for the
	lasso, the number of selected variables is easy to control by \emph{pInc-DSS}, because the entire
	trace of the adaptive lasso (\ref{dss}) is available.  To evaluate predictive performance, we used a
	single 2/3-1/3 split of the data, leading to training and test sets of $249$ and $124$ observations,
	respectively.  The lasso was computed using \texttt{GLMnet} by \cite{Friedman2010}, also for
	(\ref{dss}).
	
	The four panels of Figure~\ref{lassovsBSEMHS} report the results for the 
	maximal number of predictor variables set equal to $1$, $3$, $5$, or $10$.
	The vertical axis shows the relative reduction of the MSE on the test set as compared to 
	the empty model (all $\bbeta_i = \mathbf{0}$), defined by
	\begin{equation}\label{msecomp}
	\frac{\text{MSE}_0 - \text{MSE}(m_i)}{\text{MSE}_0},
	\end{equation}
	where $\text{MSE}_0$ is the MSE of the empty model and $\text{MSE}(m_i)$ the MSE of linear model $m_i$.
	This quantity was calculated for all 99 genes in the pathway (horizontal axis), for both
	the lasso (displayed in black) and \emph{pInc-DSS} (displayed in red), large values indicating accurate prediction. 
	The results of the lasso are somewhat more `noisy', likely due to less shrinkage of the (near-)zero
	parameter estimates, and the lasso regularly performs inferior to both the empty model 
	(negative values) and \emph{pInc-DSS}, with  gene 13 an extreme case. 
	For genes with considerable signal w.r.t.\ the empty model (e.g.\ genes 61, 93 and 98), 
	\emph{pInc-DSS} explains much more of the signal than the lasso. 
	This could be explained by less shrinkage of the non-zero parameters by the horseshoe prior, which is designed
	to separate zero and nonzero values. This is illustrated in Figure~\ref{gene98} 
	for gene 98. Gene 50 is the one exception, 
	where lasso beats \emph{pInc-DSS}, in the case of selecting 3 variables.
	
	\begin{figure}
		\begin{center}
			\includegraphics[scale=0.4]{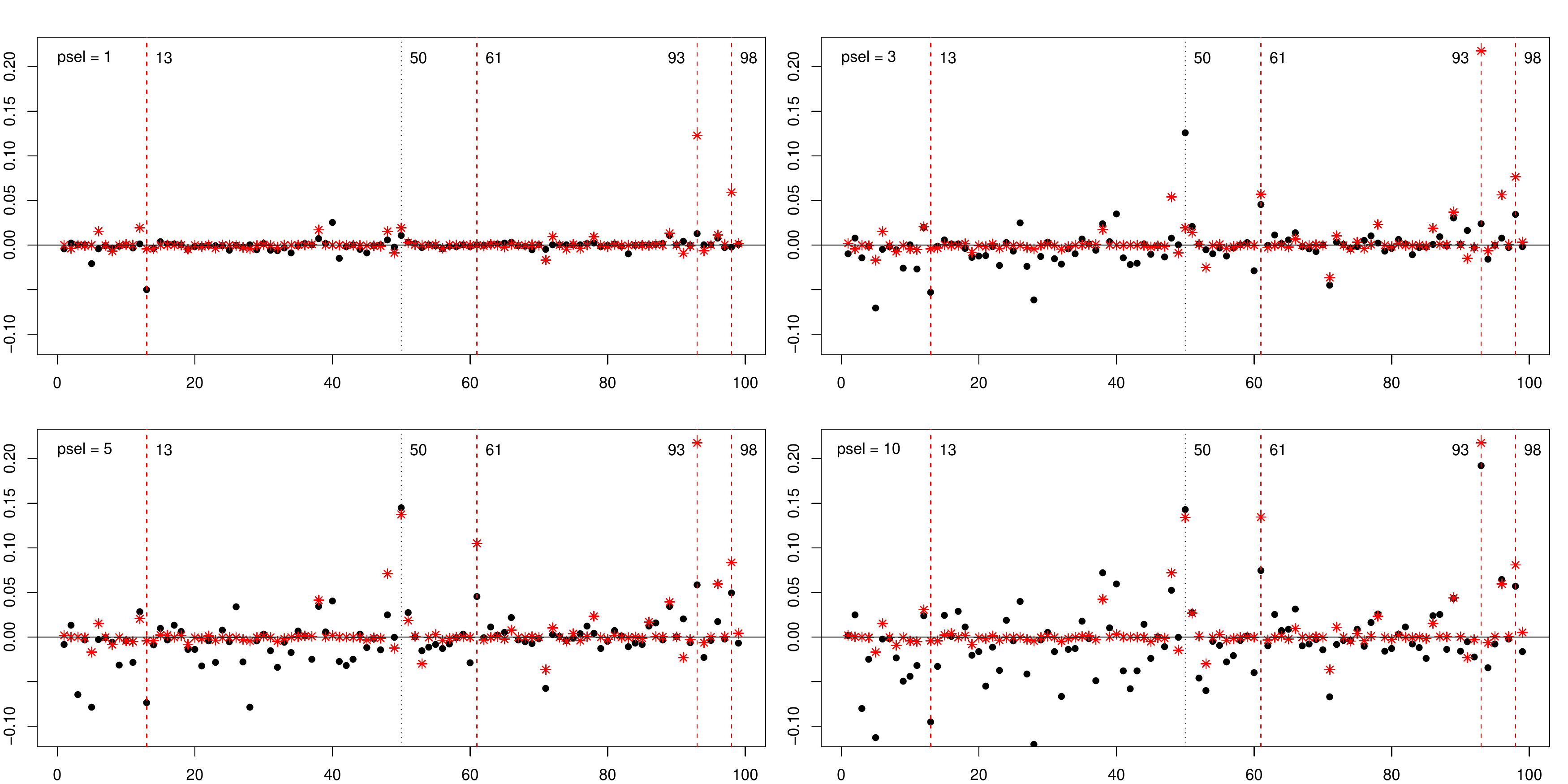}
		\end{center}
		\caption{Relative reduction of MSE (y-axis) for the lasso (black dots) 
			and \emph{pInc-DSS} (red stars) for all genes $i=1, \ldots, 99$ (x-axis) when
			maximal number of variables is fixed to 1, 3, 5, or 10 (clockwise from top-left).
			The genes with the large differences are highlighted by vertical lines}\label{lassovsBSEMHS}
	\end{figure}
	
	\begin{figure}
		\begin{center}
			\includegraphics[scale=0.4]{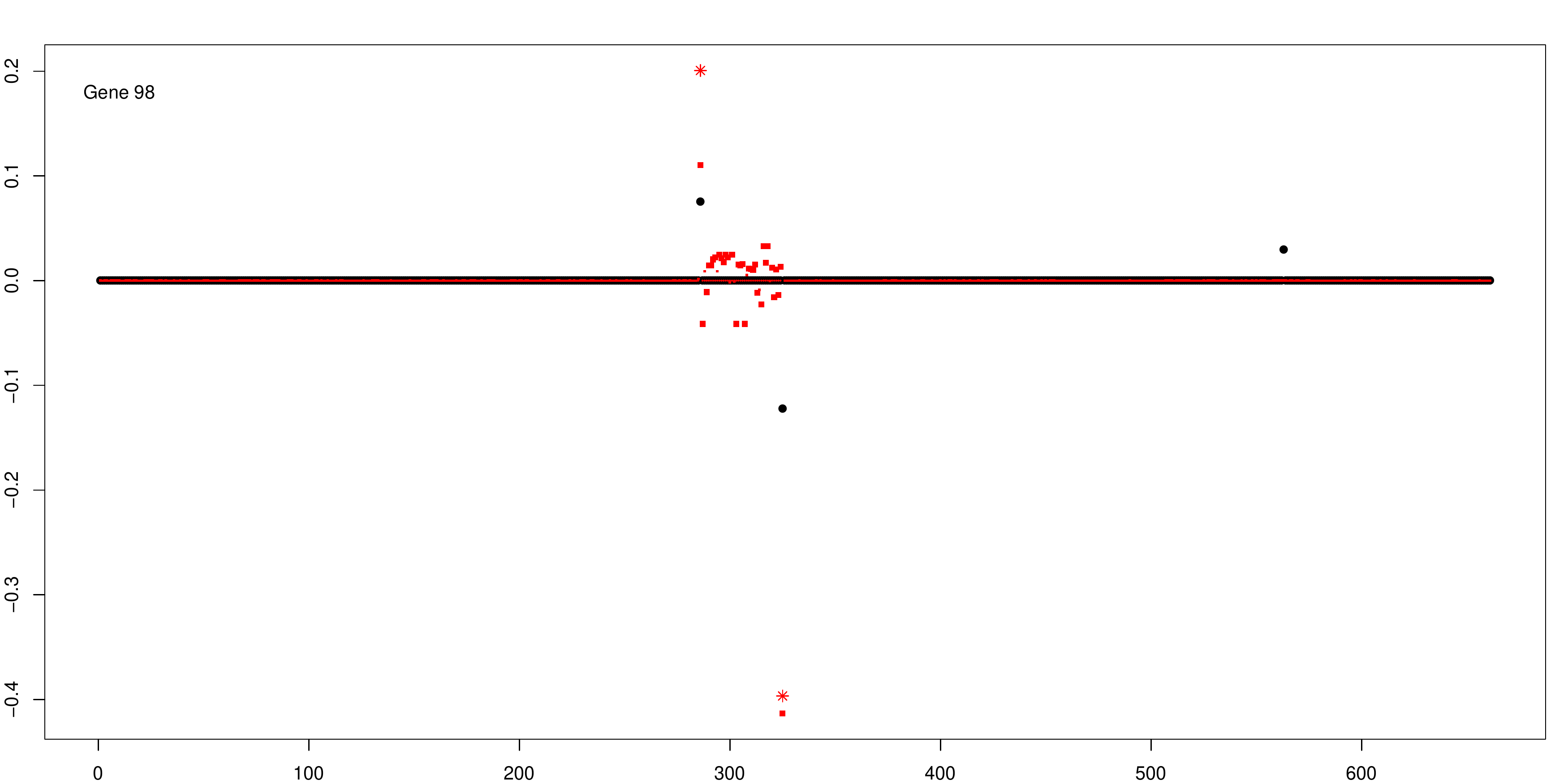}
		\end{center}
		\caption{Estimates of SNP effects on gene 98 using 
			\emph{pInc} (red squares), and \emph{pInc-DSS} (red stars) and the lasso (black dots) with 3 predictor variables
			for the latter two. X-axis denotes SNP index. }\label{gene98}
	\end{figure}

	\section{Conclusion}
	\label{end}
	
	We have introduced a sparse high-dimensional regression approach that can incorporate prior
	information on the regression parameters and can borrow information across a set of similar
	datasets.  It is based on an empirical Bayesian setup, where external information is incorporated
	through the prior, and information is borrowed across similar analyses by empirical Bayes estimation
	of hyperparameters.  We have shown the power of the approach both in model-based simulations 
	of Gaussian graphical models and in real data analyses in genomics. Incorporating the information was shown to
	enhance the analysis, even when the prior information was only partly correct (e.g.\ 50 \% accurate). 
	We explain this by the fact that the empirical Bayesian approach 
	is able to incorporate prior information in a soft manner. Such a flexible approach is particularly
	attractive in high-dimensional situations where the amount of data is  small relative to the number of parameters
	and an increasing amount of prior information is available. 
	
	To make our approach scalable to large models and/or datasets we developed a variational Bayes
	approximation to the posterior distribution resulting from the horseshoe prior distribution. We
	showed the accuracy of the resulting approximation to the marginal posterior distributions of the
	regression parameters by comparison to state-of-the-art MCMC schemes for the horseshoe prior. The
	variational Bayes approach obtained the same (if not better) accuracy at a fraction of CPU time.
	
	We studied two versions of the model, one with a gamma prior on the `sparsity' parameters and one in
	which these parameters are estimated by the empirical Bayes method. We found that the gamma prior is
	preferable when relevant prior knowledge can be used, but in the absence of prior knowledge the
	alternative model may be preferable.

	\bibliographystyle{plain}
	\bibliography{Cbib}	
\end{document}